\documentclass[final,3p,times,11pt]{elsarticle}

\usepackage{amssymb}
\usepackage{amsthm}

\usepackage{amsmath}
\usepackage{mathrsfs}
\usepackage{mathtools}
\usepackage{siunitx}
\usepackage{float}
\usepackage[labelfont=bf,labelsep=period]{caption}
\usepackage{subcaption}
\usepackage{booktabs}
\usepackage{threeparttable}
\usepackage{graphicx}
\graphicspath{{figures/}}

\usepackage{url}
\usepackage{hyperref}

\newtheorem{thm}{Theorem}
\newtheorem{lem}[thm]{Lemma}
\newdefinition{rmk}{Remark}
\usepackage{xpatch}
\xpatchcmd{\proof}{\itshape}{\prooflabelfont}{}{}
\newcommand{\prooflabelfont}{\bfseries}

\journal{Mathematical Biosciences}

\begin{document}

	\begin{frontmatter}

		\title{%
			Single-species population models with\\
			age structure and psychological effect in a polluted environment
		}

		\author[a]{Jiawei Wang}

		\author[a,b]{Ruiyang Zhou}

		\author[a,c]{Fengying Wei\corref{cor1}}
		\cortext[cor1]{Corresponding author}
		\ead{weifengying@fzu.edu.cn}

		\affiliation[a]{
			organization={School of Mathematics and Statistics, Fuzhou University},
			city={Fuzhou},
			postcode={350116},
			state={Fujian},
			country={China}
		}

		\affiliation[b]{
			organization={Institute of Fundamental and Frontier Sciences, University of Electronic Science and Technology of China},
			city={Chengdu},
			postcode={611731},
			state={Sichuan},
			country={China}
		}

		\affiliation[c]{
			organization={Key Laboratory of Operations Research and Control of Universities in Fujian, Fuzhou University},
			city={Fuzhou},
			postcode={350116},
			state={Fujian},
			country={China}
		}

		\begin{abstract}
			This paper considers a single-population model with age structure and psychological effects in a polluted environment.
			We divide the single population into two stages of larval and adult structure.
			The model uses Logistic input, and the larvae are converted into adult bodies by constant ratio.
			We only consider adulthood.
			The role of psychological effects makes the contact between adult and environmental toxins a functional form, while the contact between larvae and environmental toxins is linear.

			For the deterministic model embodied as a nonlinear time-varying system, we discuss the asymptotic stability of the system by Lyapunov one-time approximation theory, and give a sufficient condition for stability to be established.

			Considering that the contact rate between biological and environmental toxins in nature is not always constant, we make the contact rate interfere with white noise, and then modify the contact rate into a stochastic process, thus establishing a corresponding random single-population model.
			According to Itô formula and Lyapunov in the function method, we first prove the existence of globally unique positive solutions for stochastic models under arbitrary initial conditions, and then give sufficient conditions for weak average long-term survival and random long-term survival for single populations in the expected sense.
		\end{abstract}

		\begin{keyword}
			polluted environment \sep
			stochastic single-species population models \sep
			age structure \sep
			psychological effect \sep
			stability \sep
			persistence
		\end{keyword}

	\end{frontmatter}



	\section{Introduction}
	\label{sec:introduction}


	With the development of society, the problem of environmental pollution caused by human activities and industrial production has become increasingly serious, environmental pollution has always been the most threatening social and ecological problem.
	All kinds of pollutants discharged into the environment have a great impact on the normal survival of all kinds of organisms in the environment.
	All kinds of organisms, including human beings, are threatened by all kinds of poisons in the environment.
	At present, many scholars have established mathematical models to describe it.

	Since Hallam et al.~\cite{RN01,RN02,RN03,RN04,RN05}, scholars have paid more and more attention to this problem.
	In 1983, Hallam et al.~\cite{RN01} first proposed a classical deterministic mathematical model for the survival of a single population in a polluted environment.
	Other scholars continue to study the environmental pollution model~\cite{RN06,RN07,RN08,RN09,RN10,RN12}.
	Srinivasu~\cite{RN13} considered the impact of biological death on the concentration of environmental toxins, studied the problem of population's long-term survival and extinction, and considered to ensure the population's long-term survival by increasing the amount of environmental toxins removed.

	The growth of species is basically accompanied by a process of development, i.e. from adulthood to old age, from immature to mature, from young individual to adult individual.
	Each stage of the growth process often shows different characteristics, such as the young individual does not have the ability of reproduction and predation.
	Due to the limitation of intelligence development and survival experience, the young individual often does not have the ability of self avoiding danger.
	Compared with other stages of population, the ability of survival and resources competition of young individuals is weak.
	Young individuals are often easy to die, and unable to complete large-scale migration in space.

	The adult population has the ability of reproduction and predation, a high level of intellectual development.
	It is also superior in survival experience and can actively avoid danger, thus has a strong survival ability and the ability to compete with other populations for limited survival resources.
	Some of them experience that one side is injured or even killed in order to compete for spouse.
	There are significant differences in the behavior of the same stage.
	In addition, there is always a relationship of mutual transformation between organisms in different stages, which will have an impact on the extinction and long-term survival of organisms.
	Therefore, it is of practical significance to study the population model with stage structure, which has also attracted the attention of many mathematicians and biologists.

	The related work of stage structure population model can be traced back to the three-stage structure self feeding model proposed by Landahl and Hansen~\cite{RN17}, in the 1970s, and the two-stage structure stochastic model proposed by Tognetti~\cite{RN18}.
	Later, many scholars also proposed different stage structure models~\cite{RN20}.
	But until 1990, Aiello and Freedman~\cite{RN19} proposed the single group delay model with two-stage structure.
	The related research only then truly ushers in the upsurge~\cite{RN21,RN23,RN24,RN25,RN26,RN27,RN28}.

	For vertebrates with a certain level of intelligence, when life safety is threatened by diseases, predators, environmental pollution and other external threats, there are often a series of psychological effects such as fear.
	Compared with other creatures with a low level of intelligence, psychological effects affect the daily behavior of the whole population to a certain extent.

	In recent years, many scholars have tried to make the mathematical models to depict the influence of psychological effects.
	Many of them study human groups directly~\cite{RN29,RN30,RN31,RN32,RN33,RN34,RN35}, and some of them study other creatures.
	In 2018, Kumar~\cite{RN36}, studied the prey predator model of Allee effect induced by fear.
	In 2019, Lan et al.~\cite{RN37}, considered the long-term behavior of a single population model with psychological effects under the pulse input of environmental toxins.
	Assuming that the internal growth rate of a single population is interfered by white noise, a corresponding stochastic model was established.


	In this paper, we mainly consider vertebrates with sensory organs and highly differentiated nervous system.
	Compared with organisms without this structure in nature, sensory organs can interpret information in dynamic environment.
	For vertebrates, such as fish and birds, we can detect and process information in the physical world.
	When the environment is polluted, these organisms either bear the toxin concentration of surrounding environment or escapes from the living area when the psychological effect plays a role.
	Combined with the age structure, we naturally think that the intelligence level of young individuals is lower than that of adult individuals.
	Therefore, when considering the psychological effect, we should consider the two groups separately.
	Based on the above characteristics, we will study the single population with age structure and psychological effect in the polluted environment dynamic behavior.

	In the first part of this paper, based on the previous studies, a single population model with age structure and psychological effects is established.
	The model is a nonlinear time-varying system.
	According to Lyapunov's first approximation theory, we study the related properties of the original system by studying the stability of the zero solution of the approximate system.

	In the second part of this paper, we consider that the contact rate between organisms and environmental poisons is disturbed by white noise, and then we establish a stochastic model.
	At the same time, we give the proof of the weak average persistence and stochastic persistence of the single population under certain conditions.

	\section{Stability analysis of deterministic model}
	\label{sec:ode_sec}

	\subsection{Deterministic model}
	\label{subsec:ode_model}

	A single population model with psychological effects in polluted
	environment is proposed~\cite{RN34}
	\begin{equation}
		\dot{x}(t)=x(t)(b-d)-cx^2(t)-\alpha x(t)C_o(t)-\lambda x(t)g(C_e,
		\beta), t\in[0, \infty)
	\end{equation}

	The variables of the model are $x=x(t)$, which means the number of
	population at time $t$, $C_o=C_o(t)$, the degree of individual
	toxins at time $t$, $C_e=C_e(t)$, the concentration of toxins in
	living environment at time $t$, $b, d, c$ represent the natural
	birth rate, mortality rate and biological intra-specific
	restriction factors when the toxin is zero.
	$\alpha$
	represent the response intensity of organism growth to the toxin.
	$\lambda$
	represent the contact rate between organism and environmental toxicant.
	$\beta$ refers to the inhibition factors or
	psychological effects of the population in the polluted environment.
	To some extent,  $\beta$ describes the sensitivity of
	the population in the polluted environment.
	$\alpha x(t)C_o(t)$
	describes the removal amount of endotoxin in the organism at time
	$t$.
	The toxic effect of the polluted environment is represented
	by $g(C_e, \beta)$~\cite{RN31}
	\begin{equation}
		g(C_e, \beta)=\frac{C_e}{1+\beta C_e^2}
	\end{equation}
	also known as contact rate or psychological effect.
	Age structure
	is introduced, and the specific model is as follows
	\begin{equation}
		\begin{array}{lll}
			\dot{x}(t)=by(t)-d_1x(t)-\gamma x(t)-c_1x^2(t)-\alpha_1
			x(t)C_o(t)-\lambda_1 x(t)g_1(C_e(t), \beta), \\[3mm]
			\dot{y}(t)=\gamma x(t)-d_2y(t)-c_2y^2(t)-\alpha_2
			y(t)C_o(t)-\lambda_2 y(t)g_2(C_e(t), \beta),
		\end{array}
	\end{equation}

	$x, y$ refers to the number of young population and adult
	population, both of which are related to time $t$.
	In this paper,
	we stipulate that the values of $x$ and $y$ are all in the range
	of $[1, \infty)$.
	$\gamma$ represents the transformation rate from
	larva to adult, and the meaning of other symbols is consistent
	with the original.
	The toxin action in the polluted environment at
	time $t$ is indicated by $g_i(C_e, \beta), i=1,2$.
	Because there
	are differences in intelligence level and survival experience
	between larva and adult, we choose the following functions to
	describe the effect of environmental toxin respectively
	\begin{equation}
		g_1(C_e, \beta)=C_e, g_2(C_e, \beta)=\frac{C_e^p}{1+\beta C_e^q},
		1\leqslant p\leqslant q, p, q\in \mathbb{N}_+
	\end{equation}

	It is assumed that the change rule of endotoxin concentration
	meets the following requirements
	\begin{equation}
		\dot{C_o}(t)=kC_e(t)-gC_o(t)-mC_o(t)-bC_o(t)
	\end{equation}
	Among them, $kC_e(t)$ is the ratio of biological absorption of
	environmental toxins at the time of $t$.
	$gC_o(t)$ is the ratio of
	biological elimination of toxins at the time of $t$.
	$mC_o(t)$ is
	the ratio of biological purification of toxins at the time of $t$.
	$bC_o(t)$ is the ratio of toxins lost by new individuals at the
	time of $t$.
	We always assume that the environmental capacity is
	large enough.
	We also assume that the toxins in new individuals
	should be very small, and the accumulation of toxins from birth to
	adulthood is very small, so the effect of toxin concentration on
	the density of a single population can be ignored.
	The change rate
	of toxin concentration in the environment is as follows
	\begin{equation}
		\dot{C_e}(t)=u_e(t)-hC_e(t),
	\end{equation}
	$hC_e(t)$ is the ratio of toxin reduction caused by environmental
	self purification at time $t$.
	$hC_e(t)$ is the concentration of
	exogenous input toxin entering the environment at time $t$, and it
	is assumed to be a bounded non-negative differentiable function of
	time $t$.

	A single population model with age structure and psychological
	effects in the polluted environment is established:
	\begin{equation}
		\label{eq:ODEmodelFull}
		\left\{
		\begin{array}{ll}
			\dot{x}(t)=by(t)-d_1x(t)-\gamma x(t)-c_1x^2(t)-\alpha_1
			x(t)C_o(t)-\lambda_1 x(t)g_1(C_e(t), \beta), \\[3mm]
			\dot{y}(t)=\gamma x(t)-d_2y(t)-c_2y^2(t)-\alpha_2
			y(t)C_o(t)-\lambda_2 y(t)g_2(C_e(t), \beta),\\[3mm]
			\dot{C_o}(t)=kC_e(t)-gC_o(t)-mC_o(t)-bC_o(t),\\[3mm]
			\dot{C_e}(t)=u_e(t)-hC_e(t).
		\end{array}
		\right.
	\end{equation}

	The study of Bronstein\cite{RN43}, Cantalupo\cite{RN44}, Castro\cite{RN45},
	Dzieweczynskia\cite{RN46} showed that Betta splendens Regan was very
	aggressive in courtship, and the death of adult was mainly caused
	by intra-specific competition.
	Hagvar~\cite{RN47} and Uka~\cite{RN48} showed
	that the death of adult was mainly caused by intra-specific
	competition.

	Therefore, according to these biological examples, When $d_2y(t)\ll c_2y^2(t)$, so $d_2y(t)+c_2y^2(t)\approx
	c_2y^2(t)$, we modify the model~\eqref{eq:ODEmodelFull} and rewrite it as
	\begin{equation}
		\label{eq:ODEmodel}
		\left\{
		\begin{array}{lll}
			\vspace{1ex}
			\dot{x}(t)=by(t)-dx(t)-\gamma x(t)-c_1x^2(t)-\alpha_1 x(t)C_o(t)-\lambda_1 x(t) C_e(t),
			\\\vspace{1ex}
			\displaystyle \dot{y}(t)=\gamma x(t)-c_2y^2(t)-\alpha_2 y(t)C_o(t)-\frac{\lambda_2 y(t)C_e^p(t)}{1+\beta C_e^q(t)},
			\\\vspace{1ex}
			\dot{C_o}(t)=kC_e(t)-gC_o(t)-mC_o(t)-bC_o(t),
			\\\vspace{1ex}
			\dot{C_e}(t)=u_e(t)-hC_e(t).
		\end{array}
		\right.
	\end{equation}

	\subsection{Stability analysis of nonlinear time-varying system}
	\label{subsec:ode_stability}

	From~\eqref{eq:ODEmodel}, we can get
	\begin{equation}
		\label{eq:CoCe}
		\left\{
		\begin{array}{lll}
			\vspace{1ex}
			C_o(t)=\displaystyle k\int_{0}^{t}C_e(s)e^{-(g+m+b)(t-s)}\mathrm{d}s+C_o(0)e^{-(g+m+b)t},
			\\\vspace{1ex}
			C_e(t)=\displaystyle
			\int_{0}^{t}u_e(s)e^{-h(t-s)}\mathrm{d}s+C_e(0)e^{-ht}.
		\end{array}
		\right.
	\end{equation}

	Denote $ X(t)=(x(t),y(t))^{T}$,~\eqref{eq:CoCe} could be described as follows
	\begin{equation}
		\label{eq:linear system sum}
		\dot{X}(t)=A(t)X(t)+G(X(t)),
	\end{equation}
	where
	\begin{equation}
		\begin{array}{l}
			\vspace{1ex}
			A(t) =\mbox{$\left(
			\begin{array}{ccc}
				-d-\gamma-\alpha_1 C_o(t)-\lambda_1 C_e(t) & b                                                  \\
				\gamma                                     & -\alpha_2 C_o(t)-\frac{\lambda_2 C_e^p(t)}{1+\beta
				C_e^q(t)}
			\end{array}\right)$},
			\\\vspace{1ex}
			G(X(t))=\mbox{$\left(
			\begin{array}{lll}
				-c_1x^2(t) \\
				-c_2y^2(t)
			\end{array}
			\right) $}
		\end{array}
	\end{equation}

	Omitting the high-order term $G(x(t))$, obtaining the approximate
	linear system
	\begin{equation}
		\label{eq:linear system}
		\dot{X}(t)=A(t)X(t),
	\end{equation}
	linear system~\eqref{eq:linear system} is called the first order approximate equation of~\eqref{eq:linear system sum}.
	According to Lyapunov's first approximation theory~\cite{RN49}, if
	$G(X(t))$ is a higher order infinitesimal of $X(t)$ in the
	neighborhood of $X=(0, 0)^T$, the stability of~\eqref{eq:linear system sum} can often be
	studied by using the stability of linear system~\eqref{eq:linear system}.
	In fact,
	we can easily get
	\begin{equation}
		\lim_{t\to \infty}\frac{\|G(X(t))\|}{\|X(t)\|}=0,
	\end{equation}
	So next we discuss the stability of linear system~\eqref{eq:linear system}.

	\begin{thm}
		\label{eq:thm:ode_stability}
		If $|A(t)|>0$, there exists an
		$\varepsilon$ such that $|\dot{a}_{ij}(t)|\leqslant \varepsilon$,
		then the trivial solution of model~\eqref{eq:linear system} is globally
		asymptotically stable.
	\end{thm}

	\begin{proof}
		Let
		\begin{equation}
			A(t)=(a_{ij}(t))_{2\times 2}= \mbox{$\left(
			\begin{array}{ccc}
				-\Gamma_1(t) & b            \\
				\gamma       & -\Gamma_2(t)
			\end{array}
			\right) $}
		\end{equation}
		where $a_{ij}(t)$ are differentiable with $|a_{ij}(t)|\leqslant
		b\vee (d+\gamma+\alpha_1+\lambda_1)\vee (\alpha_2+\lambda_2)\vee
		\gamma \geqslant a$.
		The characteristic equation of $A(t)$ gives
		the following roots
		\begin{equation}
			\lambda(A(t))=\frac{-(\Gamma_1(t)+\Gamma_2(t))\pm\sqrt{(\Gamma_1(t)-\Gamma_2(t))^2+4b\gamma}}{2}.
		\end{equation}
		When $|A(t)|>0$, each eigenvalue of  $A(t)$ has a negative real
		part, i.e. $\mbox{Re}\lambda(A(t))<0$.
		Rewrite the characteristic
		equation as follows
		\begin{equation}
			f_A(\lambda)=\lambda^2+p_1(t)\lambda+p_2(t)=0,
		\end{equation}
		where
		\begin{equation}
			p_1(t)=-\sum_{i=1}^2 a_{ii}(t)=\Gamma_1(t)+\Gamma_2(t), \quad
			p_2(t)=|A(t)|=\Gamma_1(t)\Gamma_2(t)-b\gamma.
		\end{equation}
		By Routh-Hurwitz theorem,
		\begin{equation}
			\Delta_1(t)=p_1(t)>0, \quad \Delta_2(t)=p_1(t)p_2(t)>0.
		\end{equation}
		Let
		\begin{equation}
			x=x_1, y=x_2,
		\end{equation}
		to find $v(x_1, x_2)$ satisfying
		\begin{equation}
			\label{eq:v(x_1, x_2) 1}
			\sum_{i=1}^{2}\frac{\partial v}{\partial
			x_i}\sum_{j=1}^{2}a_{ij}x_j=-2\prod_{i=1}^{2}\Delta_i(t)\sum_{j=1}^{2}x_j^2.
		\end{equation}

		According to Barabashin formula,
		\begin{equation}
			v(x_1, x_2)=\frac{\Delta_1(t)\Delta_2(t)}{\Delta(t)}
			\left|\begin{array}{cccc}
					  0 & x_1^2 & 2x_1x_2 & x_2^2
					  \\
					  1 & a_{11} & a_{21} & 0
					  \\
					  0 & a_{12} & a_{11}+a_{22} & a_{21} \\
					  1 & 0      & a_{12}        & a_{22}
			\end{array}
			\right|,
		\end{equation}
		\begin{equation}
			\Delta(t)=
			\left|\begin{array}{cccc}
					  a_{11} & a_{21} & 0
					  \\
					  a_{12} & a_{11}+a_{22} & a_{21} \\
					  0      & a_{12}        & a_{22}
			\end{array}
			\right|,
		\end{equation}
		which could be written as
		\begin{equation}
			v(x_1, x_2)=C(t)\sum_{i,j=1}^{2}v_{ij}(t)x_ix_j,
		\end{equation}
		with $v_{ij}(t)=v_{ji}(t),
		C(t)=\frac{\Delta_1(t)\Delta_2(t)}{\Delta(t)}$, where $v_{ij}$ is
		the determinant obtained by exchanging the first column and the
		$2x_ix_j$-column ($x_i^2$-column) and then removing the first row
		and the first column in the exchanged determinant.
		Therefore
		\begin{equation}
			\label{eq:v(x_1, x_2) 2}
			v(x_1, x_2)=\sum_{i,j=1}^{2}V_{ij}(t)x_ix_j,
		\end{equation}
		with $V_{ij}(t)=C(t)v_{ij}(t)$ and $V_{ij}(t)=V_{ji}(t)$.
		Next, we
		will prove that function $v(x_1, x_2)$ is positive definite, given
		quadratic form
		\begin{equation}
			W=-\prod_{i=1}^{2}\Delta_i(t)\sum_{j=1}^{2}x_j^2
		\end{equation}

		From~\eqref{eq:v(x_1, x_2) 1}, we can find the uniquely determined $v(x_1, x_2)$, so
		the $v(x_1, x_2)$ expressed by~\eqref{eq:v(x_1, x_2) 2} should be consistent with
		the Lyapunov function in reference~\cite{RN51}.
		\begin{equation}
			\begin{array}{lll}
				v(x_1, x_2) & = & \displaystyle \Delta_2(t)\sum_{j=1}^{2}x_j^2+
				\sum_{\rho=1}^{2-1}\sum_{j=1}^{2}\prod_{s=1,s\not=\rho+1}^{2}
				\Delta_s(t)\Delta_{\rho,j}^2(x_1, x_2) \\
				& = & \displaystyle \Delta_2(t)\sum_{j=1}^{2}x_j^2+ \sum_{j=1}^{2}
				\Delta_1(t)\Delta_{1,j}^2(x_1, x_2) \\
				& = & \displaystyle \Delta_2(t)\sum_{j=1}^{2}x_j^2+
				\Delta_1(t)[(a_{22}x_1-a_{12}x_2)^2+(a_{21}x_1-a_{11}x_2)^2] \\
				& = & \displaystyle \Delta_2(t)\sum_{j=1}^{2}x_j^2+
				\Delta_1(t)[(\Gamma_1(t)x_1+bx_2)^2+(\gamma
				x_1+|\gamma_2(t)x_2)^2]
			\end{array}
		\end{equation}

		Mark $\Delta_{\rho,s}(x_1, x_2)$ can be found in reference~\cite{RN51},
		obviously
		\begin{equation}
			v(x_1, x_2) \geqslant \Delta_2(t)(x_1^2+x_2^2)
			=-\lambda_1\lambda_2(\lambda_1+\lambda_2)(x_1^2+x_2^2) \geqslant
			\delta^3(x_1^2+x_2^2)
		\end{equation}
		Therefore, the above function $v(x_1, x_2)$ is positive definite,
		so we can take the positive definite function $v(x_1, x_2)$  as
		the Lyapunov function of the system~\eqref{eq:linear system}.
		\begin{equation}
			\left(\frac{\mathrm{d}v}{\mathrm{d}t}\right)_{\eqref{eq:linear system}} =
			-2\Delta_1(t)\Delta_2(t)\sum_{j=1}^{2}x_j^2+
			\sum_{i,j=1}^{2}\dot{V}_{ij}(t)x_ix_j
		\end{equation}
		where
		\begin{equation}
			\sum_{i,j=1}^{2}\dot{V}_{ij}(t)x_ix_j \leqslant
			\sum_{i,j=1}^{2}|\dot{V}_{ij}(t)|x_ix_j \leqslant
			\frac{1}{2}\sum_{i,j=1}^{2}|\dot{V}_{ij}(t)|(x_i^2+x_j^2)
			=\sum_{i=1}^{2}\sum_{i=1}^{2}|\dot{V}_{ij}(t)|x_i^2
		\end{equation}
		Because $V_{ij}(t)=C(t)v_{ij}(t)$, so
		\begin{equation}
			\dot{V}_{ij}(t)=\dot{C}(t)v_{ij}(t)+C(t)\dot{v}_{ij}(t)=
			\Big(\sum_{i,j=1}^{2}\frac{\partial C}{\partial
			a_{ij}}\dot{a}_{ij}\Big)v_{ij}(t)+C(t)\dot{v}_{ij}(t)
		\end{equation}
		Let $|\dot{a}_{ij}(t)|\leqslant \varepsilon$, then
		\begin{equation}
			\label{eq:dotVijInequation}
			|\dot{V}_{ij}(t)|\leqslant
			\varepsilon\Big(\sum_{i,j=1}^{2}\frac{\partial C}{\partial
			a_{ij}}\Big)|v_{ij}(t)|+\varepsilon |C(t)|P_{ij}
		\end{equation}
		where $P_{ij}$ is the sum of the absolute values of the algebraic
		cofactors of the elements whose derivative is not zero in the
		determinant of $V_{ij}$ (if an element appears in the form of
		$a_{ii}+a_{jj}$, multiply by 2 before the absolute value of its
		algebraic cofactors), which is obtained from~\eqref{eq:dotVijInequation}
		\begin{equation}
			|\dot{V}_{ij}(t)|\leqslant \varepsilon D(t)|v_{ij}(t)|+\varepsilon
			|C(t)|P_{ij}(t).
		\end{equation}

		Among them $\sum_{i,j=1}^{2}\frac{\partial C}{\partial
		a_{ij}}=D(t)$, so
		\begin{equation}
			\left(\frac{\mathrm{d}v}{\mathrm{d}t}\right)_{\eqref{eq:linear system}} \leqslant
			-2\Delta_1(t)\Delta_2(t)\sum_{j=1}^{2}x_j^2+
			\varepsilon\Big\{\Big[D(t)Q_1(t)+|C(t)|P_1(t)\Big]x_1^2
			+D(t)Q_2(t)+|C(t)|P_2(t)\Big]x_2^2\Big\},
		\end{equation}
		where
		\begin{equation}
			Q_1(t)=\sum_{j=1}^{2}|V_{1j}(t)|,
			Q_2(t)=\sum_{j=1}^{2}|V_{2j}(t)|,
			P_1(t)=\sum_{j=1}^{2}P_{1j}(t),
			P_2(t)=\sum_{j=1}^{2}P_{2j}(t),
		\end{equation}

		Take
		\begin{equation}
			\varepsilon=\min\Big\{\frac{\Delta_1(t)\Delta_2(t)}{D(t)Q_1+|C(t)|P_1},
			\frac{\Delta_1(t)\Delta_2(t)}{D(t)Q_2+|C(t)|P_2} \Big\}>0
		\end{equation}
		so when $|\dot{a}_{ij}(t)|\leqslant \varepsilon$, there are
		\begin{equation}
			\left(\frac{\mathrm{d}v}{\mathrm{d}t}\right)_{\eqref{eq:linear system}} \leqslant
			-2\Delta_1(t)\Delta_2(t)\sum_{j=1}^{2} x_j^2
		\end{equation}

		By using the condition $\mbox{Re}\lambda(A(t))\leqslant -\delta<0$
		\begin{equation}
			\left(\frac{\mathrm{d}v}{\mathrm{d}t}\right)_{\eqref{eq:linear system}} \leqslant -K^*\sum_{j=1}^{2}x_j^2
		\end{equation}
		where $K^*$ is a constant independent of $t$, so $ \left(\frac{\mathrm{d}v}{\mathrm{d}t}\right)_{\eqref{eq:linear system}} $ is
		negative definite.

		Since $V(x_1, x_2)$) is positive definite and
		has an infinitesimal upper limit (obtained immediately from
		$|a_{ij}(t)|\leqslant a$), combined with the conclusion proved
		previously, the zero solution of the system~\eqref{eq:linear system} is
		asymptotically stable.
		Because the system~\eqref{eq:linear system} is
		linear, the stability has global properties.
	\end{proof}

	\section{Survival analysis of stochastic single population model}
	\label{sec:sde_sec}

	\subsection{Stochastic Model}
	\label{subsec:sde_model}

	Because the contact rate between a single population and
	environmental toxins is inevitably affected by other conditions,
	such as weather conditions, temperature and other nearby noises,
	it is natural to think that constant contact rate $\lambda_i$ is
	replaced by random variable $\lambda_i+\sigma_i\xi_i(t)$, where
	$\xi_i(t)$ is white noise, and $\xi_i(t)=\frac{dB_i(t)}{dt}$,
	$B_i(t)$ are two independent standard Brown motions defined on
	$(\Omega, \mathcal{F}, P)$.
	Therefore, we get random single
	species with psychological effects group model
	\begin{equation}
		\label{eq:SDEmodel}
		\left\{
		\begin{array}{lll}
			\vspace{1ex}
			\displaystyle \mathrm{d}x(t)=[by(t)-dx(t)-\gamma
			x(t)-c_1x^2(t)-\alpha_1 x(t)C_o(t)-\lambda_1 x(t)
			C_e(t)]\mathrm{d}t-\sigma_1 x(t)C_e(t)\mathrm{d}B_1(t),
			\\\vspace{1ex}
			\displaystyle \mathrm{d}y(t)=\Big[\gamma x(t)-c_2y^2(t)-\alpha_2
			y(t)C_o(t)-\frac{\lambda_2 y(t)C_e^p(t)}{1+\beta C_e^q(t)}\Big]-
			\frac{\sigma_2 y(t)C_e^p(t)}{1+\beta
			C_e^q(t)}\mathrm{d}B_2(t),
			\\\vspace{1ex}
			\displaystyle
			\mathrm{d}C_o(t)=[kC_e(t)-gC_o(t)-mC_o(t)-bC_o(t)]\mathrm{d}t,\\[3mm]
			\displaystyle \mathrm{d}C_e(t)=[u_e(t)-hC_e(t)]\mathrm{d}t.
		\end{array}
		\right.
	\end{equation}
	For the convenience of subsequent derivation, we propose the
	following lemmas:
	\begin{lem}
		\label{lem:unique_eep}
		Under the condition of non-pollution,
		the system~\eqref{eq:SDEmodel} has a unique equilibrium point $ (A, B)
		$ in the first quadrant.
	\end{lem}

	\begin{proof}
		Under the condition of non-pollution, the
		system~\eqref{eq:SDEmodel} is simplified as
		\begin{equation}
			\left\{
			\begin{array}{l}
				\dot{x}(t)=by(t)-dx(t)-\gamma x(t)-c_1x^2(t),
				\\
				\dot{y}(t)=\gamma x(t)-c_2y^2(t),
			\end{array}
			\right.
		\end{equation}

		The equilibrium point of the system is required, that is, to find
		the solution of the following equations
		\begin{equation}
			\label{eq:EquationGroup}
			\left\{
			\begin{array}{l}
				0=by(t)-dx(t)-\gamma x(t)-c_1x^2(t),
				\\
				0=\gamma x(t)-c_2y^2(t),
			\end{array}
			\right.
		\end{equation}

		Since the parameters of the system are all positive, there is only
		one solution in the first quadrant of the equation group, which is
		set as $ (A, B) $.
		For the convenience of writing, we write the
		equation group~\eqref{eq:EquationGroup} in the following form:
		\begin{equation}
			\label{eq:RewrittenEquationGroup}
			\left\{
			\begin{array}{l}
				x=\tilde{a}y^2,
				\\
				y=\tilde{b}x+\tilde{c}x^2,
			\end{array}
			\right.
		\end{equation}
		of which
		\begin{equation}
			\tilde{a}=\frac{c_2}{\gamma},
			\tilde{b}=\frac{d+\gamma}{b},
			\tilde{c}=\frac{c_1}{b}.
		\end{equation}
		The solution of equations~\eqref{eq:RewrittenEquationGroup} in the first quadrant is $ (A, B) $,
		where
		\begin{equation}
			A=\tilde{a}B^2, B=\Lambda-\frac{\tilde{b}}{3\tilde{a}\tilde{c}},
			\Lambda=\Big(\frac{\sqrt{27\tilde{c}+4\tilde{a}\tilde{b}^3}}{2\sqrt{27}\tilde{a}^2\tilde{c}^{\frac{3}{2}}}
			+\frac{1}{2\tilde{a}^2\tilde{c}}\Big) ^{\frac{1}{3}}
		\end{equation}
	\end{proof}

	\begin{lem}
		For models~\eqref{eq:ODEmodel} and~\eqref{eq:SDEmodel},
		if $k\leqslant g+m+b$, $u_e^*\leqslant h$, then,
		$0\leqslant C_o(t)<1, 0\leqslant C_e(t)<1$ for $t\in
		\mathbb{R}_+$.
	\end{lem}

	\begin{proof}
		The explicit solutions $C_o(t), C_e(t)$
		are as follows:
		\[
			\begin{array}{ll}
				\vspace{1ex}
				\displaystyle C_o(t)=k\int_{0}^{t}C_e(s)e^{-(g+m+b)(t-s)}\mathrm{d}s+C_o(0)e^{-(g+m+b)t},\\
				\vspace{1ex}
				\displaystyle C_e(t)=\int_{0}^{t}u_e(s)e^{-h(t-s)}\mathrm{d}s+C_e(0)e^{-ht}.
			\end{array}
		\]
		for the given initial conditions $0\leqslant C_o(0)<1, 0\leqslant
		C_e(0)<1$, the explicit solutions can be estimated
		\[
			\begin{array}{ll}
				\vspace{1ex}
				\displaystyle C_o(t)\leqslant 1-e^{-(g+m+b)t}+C_o(0)e^{-(g+m+b)t}<1, \\
				\vspace{1ex}
				\displaystyle C_e(t)\leqslant 1-e^{-ht}+C_e(0)e^{-ht}<1.
			\end{array}
		\]
		So $0\leqslant C_o(t)<1$ and $0\leqslant C_e(t)<1$ hold.
	\end{proof}

	\subsection{Existence and uniqueness of a global positive solution}
	\label{subsec:sde_global_positive_solution}

	\begin{thm}
		For any initial value $\left(x(0),
		y(0), C_{o}(0), C_{e}(0)\right) \in \mathbb{R}_{+}^{4}$, there is
		a unique solution of model~\eqref{eq:SDEmodel}, and the solution does
		not leave $\mathbb{R}_{+}^{4}$ according to probability 1.
	\end{thm}

	\begin{proof}
		Our proof refers to the work of Mao~\cite{RN52}.
		By Lyapunov method mentioned in Lemma 2 of~\cite{RN53}, and the
		references~\cite{RN54,RN55}, the coefficients of model~\eqref{eq:SDEmodel}
		are locally Lipschitz continuous.
		On $\left[0, \tau_{e}\right)$,
		for any initial value $\left(x(0), y(0), C_{o}(0), C_{e}(0)\right)
		\in \mathbb{R}_{+}^{4}$, there is a local solution $\left(x(t),
		y(t), C_{o}(t), C_{e}(t)\right)$, where $ \tau_{e} $ is the blow
		up time.
		In order to prove that the solution is global, we need to
		prove that $ \tau_{e}=\infty $ holds almost surely.
		Let $m_{0}>1$
		be large enough, for $ m>m_{0} $, each component of the initial
		value $(x(0), y(0), C_{0}(0), C_{e}(0)) \in \mathbb{R}_{+}^{4}$ is
		in the interval $[\frac{1}{m_{0}}, m_{0}]$.
		Define the stop time
		\begin{equation}
			\begin{array}{lll}
				\vspace{1ex}
				\tau_{m} & = & \inf \left\{ t \in [ 0, \tau_{e} ):\min \left\{x(t), y(t), C_{o}(t), C_{e}(t) \right\} \leqslant \displaystyle \frac{1}{m} \right.
				\\\vspace{1ex}
				& & \left. \mbox{or }  \max \left\{x(t), y(t), C_{o}(t), C_{e}(t)\right\} \geqslant m \right\}.
			\end{array}
		\end{equation}

		In this paper, we set $\inf \phi=\infty$.
		Obviously, when $m
		\rightarrow \infty$, $\tau_{m}$ increases and $\tau_{m}<\tau_{e}$.
		We denote  $\tau_{\infty}=\lim _{m \rightarrow \infty} \tau_{m}$.
		We claim that for all $t \geqslant 0$, $\left(x(t), y(t),
		C_{o}(t), C_{e}(t)\right) \in \mathbb{R}_{+}^{4}$, and
		$\tau_{\infty}=\infty$.
		Use the counter evidence method below.
		Otherwise, for any $m>m_{0}$, there is a pair of constants $T>0$,
		$\varepsilon \in (0,1)$ such that $\mathbb{P}\left\{\tau_{m}
		\leqslant T\right\} \geqslant \varepsilon$.
		We define a
		$C^{2}$-function as follows: $V_{1}: \mathbb{R}_{+}^{4}
		\rightarrow \mathbb{R}_{+}$, and let $ x=x(t) $, $ y=y(t) $, $
		C_o=C_o(t) $, $ C_e=C_e(t) $.
		\[V_{1}(t)= (x-1-\ln x)+(y-1-\ln y)+(C_{o}-1-\ln C_{o})+(C_{e}-1-\ln C_{e}).\]

		For all $t \in\left[0, \tau_{e}\right)$, according to Itô's
		formula, together with $ V_{1}(t) \coloneqq V_{1}\left(x(t), y(t),
		C_{o}(t), C_{e}(t)\right) $
		\begin{equation}
			\label{eq:dV1}
			\begin{array}{lll}
				\mathrm{d} V_{1}(t)
				& = & \displaystyle \left(1-\frac{1}{x}\right) \mathrm{d} x+\left(1-\frac{1}{y}\right) \mathrm{d} y+\left(1-\frac{1}{C_{o}}\right) \mathrm{d} C_{o}+\left(1-\frac{1}{C_{e}}\right) \mathrm{d} C_{e}
				\\[12pt]
				& & +\displaystyle \frac{1}{2 x^{2}}(\mathrm{d} x)^{2}+\frac{1}{2 y^{2}}(\mathrm{d} y)^{2}
				\\[12pt]
				& = & \displaystyle \mathcal{L} V_{1} \mathrm{d} t-(x-1) \sigma_{1} C_{e} \mathrm{d} B_{1}(t)
				-(y-1) \frac{\sigma_{2} C_{e}^{p}}{1+\beta C_{e}^{q}} \mathrm{d}
				B_{2}(t)
			\end{array}
		\end{equation}
		where
		\begin{equation}
			\label{eq:LV1}
			\begin{array}{lll}
				\mathcal{L} V_{1}(t)
				& = & (x-1)\left(b \frac{y}{x}-d-\gamma-c_{1} x-\alpha_{1} C_{0}(t)-\lambda_{1} C_{e}\right)+\frac{1}{2} \sigma_{1}^{2} C_{e}^{2}
				\\[12pt]
				& & \displaystyle +(y-1)\left(\gamma \frac{x}{y}-c_{2} y-\alpha_{2} C_o-\frac{\lambda_{2} C_{e}^{p}}{1+\beta C_{e}^{q}}\right)+\frac{\sigma_{2} C_{e}^{2 p}}{2\left(1+\beta C_{e}^{q}\right)^{2}}
				\\[12pt]
				& & \displaystyle +k C_{e}-(g+m+b) C_o-k \frac{C_{e}}{C_o}+g+m+b+u_{e}-h C_{e}-\frac{u_{e}}{C_{e}}+h
				\\[12pt]
				& < & \displaystyle \left(c_{1}-d\right) x-c_{1} x^{2}+\left(b+c_{2}\right) y-c_{2} y^{2} +\left(\alpha_{1}+\alpha_{2}\right)+\left(\lambda_{1}+\lambda_{2}\right)+\frac{\sigma_{1}^{2}+\sigma_{2}^{2}}{2}
				\\[12pt]
				& & +\gamma+k+g+m+b+d+\left(u_{e}\right)^{*}+h
				\\[12pt]
				& \leqslant & \displaystyle \frac{\left(c_{1}-d\right)^{2}}{4 c_{1}}+\frac{\left(b+c_{2}\right)^{2}}{4 c_{2}}+\left(\alpha_{1}+\alpha_{2}\right)+\left(\lambda_{1}+\lambda_{2}\right)+\frac{\sigma_{1}^{2}+\sigma_{2}^{2}}{2}
				\\[12pt]
				& & +\gamma+k+g+m+b+d+\left(u_{e}\right)^{*}+h
				\\[12pt]
				& \coloneqq & M>0.
			\end{array}
		\end{equation}

		Hence it may be concluded that
		\begin{equation}
			\label{eq:dV1inequation}
			\mathrm{d} V_{1}(t)<M \mathrm{d} t-(x-1) \sigma_{1} C_{e} \mathrm{d} B_{1}(t)
			-(y-1) \frac{\sigma_{2} C_{e}^{p}}{1+\beta C_{e}^{q}} \mathrm{d} B_{2}(t).
		\end{equation}

		Integrate the two sides of~\eqref{eq:dV1inequation} from 0 to $\tau_{m} \wedge T$, and then take the expectation to obtain the following inequality
		\begin{equation}
			\label{eq:EV1}
			\begin{array}{lll}
				\displaystyle \mathbb{E}\left(V_{1}(\tau_{m} \wedge T\right)
				& < {} & \displaystyle V_{1}(0)+\mathbb{E} \int_{0}^{\tau_{m} \wedge T} M \mathrm{d} s
				\\[12pt]
				& & \displaystyle -\mathbb{E} \int_{0}^{\tau_{m} \wedge T}(x(s)-1) \sigma_{1} C_{e}(s) \mathrm{d} B_{1}(s)
				\\[12pt]
				& & \displaystyle -\mathbb{E} \int_{0}^{\tau_{m} \wedge T}(y(s)-1) \frac{\sigma_{2} C_{e}^{p}(s)}{1+\beta C_{e}^{q}(s)} \mathrm{d} B_{2}(s)
				\\[12pt]
				& \leqslant {} & V_{1}(0)+M T.
			\end{array}
		\end{equation}

		Note that $\Omega_{m}=\left\{\tau_{m} \wedge T\right\}$, then
		$\mathbb{P}\left\{\Omega_{m}\right\} \geqslant \varepsilon$.
		For
		any $w \in \Omega_{m}$, by the definition of stop time, there is
		\begin{equation}
			x\left(\tau_{m}, \omega\right) \wedge y\left(\tau_{m}, \omega\right) \wedge C_{o}\left(\tau_{m}, \omega\right) \wedge C_{e}\left(\tau_{m}, \omega\right)=\frac{1}{m},
		\end{equation}
		or
		\begin{equation}
			x\left(\tau_{m}, \omega\right) \vee y\left(\tau_{m}, \omega\right) \vee C_{o}\left(\tau_{m}, \omega\right) \vee C_{e}\left(\tau_{m}, \omega\right)=m.
		\end{equation}

		Therefore, there are
		\begin{equation}
			V_{1}\left(\tau_{m}, \omega\right) \geqslant \min \left\{m-1-\ln m, \frac{1}{m}-1+\ln m\right\}
		\end{equation}

		Then inequality~\eqref{eq:EV1} is equivalent to
		\begin{equation}
			\begin{array}{lll}
				V_{1}(0)+M T
				& \geqslant {} & \mathbb{E}\left[1_{\Omega_{m}}(\omega) \cdot V_{1}\left(\tau_{m} \wedge T\right)\right]
				\\[12pt]
				& \geqslant {} & \mathbb{P}\left\{\Omega_{m}\right\} \cdot \min \left\{m-1-\ln m, \frac{1}{m}-1+\ln m\right\}
				\\[12pt]
				& \geqslant {} & \varepsilon \min \left\{m-1-\ln m, \frac{1}{m}-1+\ln m\right\}.
			\end{array}
		\end{equation}
		where $1_{\Omega_{m}}(\omega)$ is the demonstrative function of $\Omega_{m}$.
		Let $m \rightarrow \infty$, we can get
		\begin{equation}
			+\infty>V_{1}(0)+M T \geqslant+\infty.
		\end{equation}

		The contradiction arises, therefore,
		$\mathbb{P}\left\{\tau_{m}=\infty\right\}=1$.
	\end{proof}

	\subsection{Weakly persistent in the mean}
	\label{subsec:weakly_persistent}

	\begin{thm}
		\label{thm:weakly_persistent}
		Let $X(0)=\left(x(0), y(0), C_{o}(0),
		C_{e}(0)\right)$ is the initial value, $ X(t) $ be a solution of
		model~\eqref{eq:SDEmodel}.
		If the coefficients of~\eqref{eq:SDEmodel}
		satisfy
		\begin{equation}
			\label{eq:conTh3}
			\begin{array}{lll}
				\vspace{1ex}
				\displaystyle \sigma_{1}^{2}<\frac{b B}{A}+ c_{1}
				A-\frac{b}{2}-\frac{\gamma}{2},
				&
				\displaystyle \sigma_{2}^{2}<\frac{\gamma A}{B}+ c_{2}B-\frac{b}{2}-\frac{\gamma}{2},
				\\\vspace{1ex}
				\displaystyle k<2(g+m+b),
				&
				\displaystyle \sigma_{1}^{2} A^{2}+\sigma_{2}^{2} B^{2}<2 h-k-1.
			\end{array}
		\end{equation}
		Then the solution $ X(t) $ of~\eqref{eq:SDEmodel} has the properties:
		\begin{equation}
			\limsup _{t \rightarrow \infty} \frac{1}{t} \mathbb{E}
			\int_{0}^{t}\left(\eta_{1}(x(s)-A)^{2}+\eta_{2}(y(s)-B)^{2}+\eta_{3} C_{o}^{2}(s)+\eta_{4} C_{e}^{2}(s)\right) \mathrm{d} s \leqslant
			K.
		\end{equation}
		In other words, in the long run, the single population near the
		pollution-free equilibrium point $(A, B, 0, 0)$ is weak persistent
		in the mean.
	\end{thm}

	\begin{proof}
		We set new variables
		\begin{equation}
			u_{1}(t)=x(t)-A,\quad u_{2}(t)=y(t)-B,\quad v(t)=C_{o}(t),\quad
			w(t)=C_{e}(t),
		\end{equation}
		where the expressions of $ A $ and $ B $ are given in Lemma~\ref{lem:unique_eep}.
		Therefore, the model~\eqref{eq:SDEmodel} is rewritten as follows:
		\begin{equation}
			\left\{
			\begin{array}{l}
				\vspace{1ex}
				\displaystyle \mathrm{d} u_{1}=\left(u_{1}+A\right)\left(\frac{b\left(u_{2}+B\right)}{u_{1}+A}-\frac{b B}{A}-c_{1} u_{1}-\alpha_{1} v-\lambda_{1} w\right)
				\mathrm{d} t-\left(u_{1}+A\right) \sigma_{1} w \mathrm{d} B_{1}(t) ,
				\\\vspace{1ex}
				\displaystyle \mathrm{d}u_{2}=\left(u_{2}+B\right)\left(\frac{\gamma\left(u_{1}+A\right)}{u_{2}+B}
				-\frac{\gamma A}{B}-c_{2} u_{2}-\alpha_{2} v-\frac{\lambda_{2} w^{p}}{1+\beta w^{q}}\right) \mathrm{d} t
				-\left(u_{2}+B\right) \frac{\sigma_{2} w^{p}}{1+\beta w^{q}} \mathrm{d}
				B_{2}(t),
				\\\vspace{1ex}
				\displaystyle \mathrm{d} v=[k u-(g+m+b) v] \mathrm{d} t ,
				\\\vspace{1ex}
				\displaystyle \mathrm{d} w=\left[u_{e}-h w\right] \mathrm{d} t.
			\end{array}
			\right.
		\end{equation}

		We define a $C^{2}$-function
		$V_{2}=u_{1}^{2}+u_{2}^{2}+v^{2}+w^{2}.
		$ From It\^{o}'s formula
		\begin{equation}
			\label{eq:dV2}
			\mathrm{d} V_{2}
			= \mathcal{L} V_{2}\mathrm{d} t
			-2\left(u_{1}+A\right) \sigma_{1} u_{1} w \mathrm{d} B_{1}(t)
			-2\left(u_{2}+B\right) \frac{\sigma_{2} u_{2} w^{p}}{1+\beta w^{q}} \mathrm{d} B_{2}(t) ,
		\end{equation}
		by the elementary equality $2xy\leqslant x^2+y^2$ for $x>0$ and
		$y>0$, we get
		\begin{equation}
			\label{eq:LV2}
			\begin{array}{lll}
				\mathcal{L} V_{2}
				& = & 2 u_{1}\left(u_{1}+A\right)\left(\frac{b\left(u_{2}+B\right)}{u_{1}+A}-\frac{b B}{A}-c_{1} u_{1}-\alpha_{1} v-\lambda_{1} w\right)+\left(u_{1}+A\right)^{2} \sigma_{1}^{2} w^{2}
				\\[12pt]
				& & +2 u_{2}\left(u_{2}+B\right)\left(\frac{\gamma\left(u_{1}+A\right)}{u_{2}+B}-\frac{\gamma A}{B}-c_{2} u_{2}-\alpha_{2} v-\frac{\lambda_{2} w^{p}}{1+\beta w^{q}}\right)
				\\[12pt]
				& & +\left(u_{2}+B\right)^{2} \frac{\sigma_{2}^{2} w^{2 p}}{\left(1+\beta w^{q}\right)^{2}}+2 v[k w-(g+m+b) v]+2 w\left(u_{e}-h w\right)
				\\[12pt]
				& < & 2 b u_{1} u_{2}-\frac{2 b B}{A} u_{1}^{2}+2 \gamma u_{1} u_{2}-\frac{2 \gamma A}{B} u_{2}^{2}
				-2 c_{1} A u_{1}^{2}+\sigma_{1}^{2} u_{1}^{2}\\[12pt]
				& & +2 \sigma_{1}^{2} A u_{1} w^{2}+\sigma_{1}^{2} A^{2} w^{2}-2 c_{2} B u_{2}^{2}+\sigma_{2}^{2} u_{2}^{2}+2 \sigma_{2}^{2} B u_{2} w^{2}+\sigma_{2}^{2} B^{2} w^{2}
				\\[12pt]
				& & +k w^{2}+k v^{2}-2(g+m+b) v^{2}+w^{2}+u_{e}^{2}-2 h w^{2}
				\\[12pt]
				& \leqslant &
				-\eta_{1} u_{1}^{2}-\eta_{2} u_{2}^{2}-\eta_{3} v^{2}-\eta_{4}
				w^{2}+K,
			\end{array}
		\end{equation}
		where
		\begin{equation}
			\begin{array}{lll}
				\eta_{1}=\frac{2 b B}{A}+2 c_{1} A-2\sigma_{1}^{2}-b-\gamma,
				&
				\eta_{2}=\frac{2 \gamma A}{B}+2 c_{2}
				B-2\sigma_{2}^{2}-b-\gamma,
				\\[3mm]
				\eta_{3}=2(g+m+b)-k,
				&
				\eta_{4}=2 h-\sigma_{1}^{2} A^{2}-\sigma_{1}^{2} B^{2}-k-1,
				\\[3mm]
				K=\sigma_{1}^{2} A^{2}+\sigma_{2}^{2} B^{2}+\left(u_{e}^{*}\right)^{2}.
			\end{array}
		\end{equation}
		Integrating from $ 0 $ to $ t $ on both sides of~\eqref{eq:dV2} and
		taking the expectation
		\begin{equation}
			\label{eq:EV2}
			\mathbb{E} \int_{0}^{t}\left(\eta_{1} u_{1}^{2}(s)+\eta_{2} u_{2}^{2}(s)+\eta_{3} v^{2}(s)+\eta_{4} w^{2}(s)\right) \mathrm{d} s
			<\mathbb{E}(V_{2}(0)-V_{2}(t))+K t \leqslant K t.
		\end{equation}
		Further, we obtain
		\begin{equation}
			\limsup _{t \rightarrow \infty} \frac{1}{t} \mathbb{E}
			\int_{0}^{t}\left(\eta_{1} u_{1}^{2}(s)+\eta_{2} u_{2}^{2}(s)+\eta_{3} v^{2}(s)+\eta_{4} w^{2}(s)\right) \mathrm{d} s \leqslant K.
		\end{equation}
	\end{proof}

	\subsection{Stochastic permanence}
	\label{subsec:stochastic_permanence}

	\begin{thm}
		\label{thm:stochastic_permanence}
		If the coefficients of the model satisfy
		\begin{equation}
			\label{eq:stochastic_permanence}
			\limsup _{t \rightarrow \infty}\left(\left(\alpha_{1}+\alpha_{2}\right) C_{o}(t)+\lambda_{1} C_{e}(t)+\frac{\lambda_{2} C_{e}^{p}(t)}{1+\beta C_{e}^{q}(t)}+3 \sigma^{2} C_{e}^{2}(t)\right)
			\leqslant 4 \sqrt{b \gamma}-\frac{\eta}{\theta}-\frac{c_{1}^{2}+c_{2}^{2}}{2}-d-\gamma,
		\end{equation}
		where $\sigma=\max \{\sigma_{1}, \sigma_{2}\}$, then the total
		number of the juveniles and the adults is stochastic permanence.
	\end{thm}

	\begin{proof}
		We choose positive constants $\theta$ and
		$\eta$ such that $\theta\geqslant 1$, $\eta \geqslant 1$, and
		define $C^{2}$-function $V_{3}: \mathbb{R}_{+}^{4} \rightarrow
		\mathbb{R}_{+}$ as follows:
		\begin{equation}
			V_{3}=e^{\eta t}\left(1+\frac{1}{x^{2}}+\frac{1}{y^{2}}\right)^{\theta}.
		\end{equation}

		By It\^{o}'s formula, we obtain that
		\begin{equation}
			\label{eq:dV3}
			\begin{array}{lll}
				\vspace{1ex}
				\mathrm{d} V_{3}
				& = & \displaystyle \eta e^{\eta t}\left(1+\frac{1}{x^{2}}+\frac{1}{y^{2}}\right)^{\theta} \mathrm{d} t
				-2 \theta e^{\eta t}\left(1+\frac{1}{x^{2}}+\frac{1}{y^{2}}\right)^{\theta-1}
				\left(\frac{1}{x^{3}} \mathrm{d} x+\frac{1}{y^{3}} \mathrm{d}
				y\right)
				\\\vspace{1ex}
				& & \displaystyle +\theta e^{\eta t}\left[\left(1+\frac{1}{x^{2}}+\frac{1}{y^{2}}\right)^{\theta-1}
				\frac{3}{x^{4}}+(\theta-1)\left(1+\frac{1}{x^{2}}+\frac{1}{y^{2}}\right)^{\theta-2} \frac{2}{x^{6}}\right](\mathrm{d} x)^{2}
				\\\vspace{1ex}
				& & \displaystyle +\theta e^{\eta t}\left[\left(1+\frac{1}{x^{2}}+\frac{1}{y^{2}}\right)^{\theta-1}
				\frac{3}{y^{4}}+(\theta-1)\left(1+\frac{1}{x^{2}}+\frac{1}{y^{2}}\right)^{\theta-2} \frac{2}{y^{6}}\right](\mathrm{d} y)^{2}
				\\\vspace{1ex}
				& = & \displaystyle \mathcal{L} V_{3} \mathrm{d}t+M(t),
			\end{array}
		\end{equation}
		where
		\begin{equation}
			\begin{array}{lll}
				M(t) & = & \displaystyle
				2 \theta e^{\eta\,t}\left(1+\frac{1}{x^{2}}+\frac{1}{y^{2}}\right)^{\theta-1}
				\left(\frac{\sigma_{1} C_{e}}{x^{2}} \mathrm{d} B_{1}(t)
				+\frac{\sigma_{1} C_{e}^{p}}{y^{2}\left(1+\beta C_{e}^{q}\right)} \mathrm{d}
				B_{2}(t)\right),
				\\[12pt]
				\mathcal{L} V_{3}
				& = & \displaystyle
				\eta e^{\eta t}\left(1+\frac{1}{x^{2}}+\frac{1}{y^{2}}\right)^{\theta}
				\\[12pt]
				& & \displaystyle
				-2 \theta e^{\eta t}\left(1+\frac{1}{x^{2}}+\frac{1}{y^{2}}\right)^{\theta-1}
				\left[\frac{1}{x^{2}}\left(b \frac{y}{x}-d-\gamma-c_{1} x-\alpha_{1} C_{o}-\lambda_{1} C_{e}\right)\right.
				\\[12pt]
				& & \displaystyle
				\left.+\frac{1}{y^{2}}\left(\gamma \frac{x}{y}-c_{2} y-\alpha_{2} C_{o}-\frac{\lambda_{2} C_{e}^{p}}{1+\beta C_{e}^{q}}\right)\right]
				\\[12pt]
				& & \displaystyle
				+\theta e^{\eta t}\left(1+\frac{1}{x^{2}}+\frac{1}{y^{2}}\right)^{\theta-1}
				\left[\sigma_{1}^{2} C_{e}^{2}\left(\frac{3}{x^{2}}+\left(1+\frac{1}{x^{2}}+\frac{1}{y^{2}}\right)^{-1}
				\frac{2(\theta-1)}{x^{4}}\right)\right.
				\\[12pt]
				& & \displaystyle
				\left.+\frac{\sigma_{2}^{2} C_{e}^{2 p}}{\left(1+\beta C_{e}^{q}\right)^{2}}
				\left(\frac{3}{y^{2}}+\left(1+\frac{1}{x^{2}}+\frac{1}{y^{2}}\right)^{-1} \frac{2(\theta-1)}{y^{4}}\right)\right]
				\\[12pt]
				& \coloneqq & \displaystyle
				\theta e^{\eta t} \left(1+\frac{1}{x^{2}}+\frac{1}{y^{2}}\right)^{\theta-2}(I_{31}+I_{32}+I_{33}).
			\end{array}
		\end{equation}

		After simplification, we derive
		\begin{equation}
			\begin{array}{lll}
				I_{31}& =& \displaystyle
				\frac{\eta}{\theta}
				\left[1+2\left(\frac{1}{x^{2}}+\frac{1}{y^{2}}\right)+\left(\frac{1}{x^{2}}+\frac{1}{y^{2}}\right)^{2}\right],
			\end{array}
		\end{equation}
		\begin{equation}
			\begin{array}{lll}
				I_{32} & < & \displaystyle
				-2\left(1+\frac{1}{x^{2}}+\frac{1}{y^{2}}\right)
				\left[\frac{2 \sqrt{b \gamma}}{x y}-\frac{1}{x^{2}}\left(d+\gamma+\alpha_{1} C_{o}+\lambda_{1} C_{e}\right)\right.
				\\[12pt]
				& & \displaystyle \left.-\frac{1}{y^{2}}
				\left(\alpha_{2} C_{o}+\frac{\lambda_{2} C_{e}^{p}}{1+\beta C_{e}^{q}}\right)\right]
				+2\left(1+\frac{1}{x^{2}}+\frac{1}{y^{2}}\right)\left(\frac{c_{1}}{x}+\frac{c_{2}}{y}\right),
			\end{array}
		\end{equation}
		\begin{equation}
			\begin{array}{lll}
				I_{33}
				& = & \displaystyle
				3\left(1+\frac{1}{x^{2}}+\frac{1}{y^{2}}\right)
				\left(\frac{\sigma_{1}^{2} C_{e}^{2}}{x^{2}}+\frac{\sigma_{2}^{2} C_{e}^{2p}}{y^{2}(1+\beta C_{e}^{q})^2}\right)
				+2(\theta-1)\left(\frac{\sigma_{1}^{2} C_{e}^{2}}{x^{4}}+\frac{\sigma_{2}^{2} C_{e}^{2p}}{y^{2}(1+\beta
				C_{e}^{q})^2}\right)
				\\[12pt]
				& < & \displaystyle 3\sigma^{2}C_{e}^{2}\left(1+\frac{1}{x^{2}}+\frac{1}{y^{2}}\right)
				\left(\frac{1}{x^{2}}+\frac{1}{y^{2}}\right)
				+ 2(\theta-1)\sigma^{2}C_{e}^{2}
				\left(\frac{1}{x^{4}}+\frac{1}{y^{4}}\right),
			\end{array}
		\end{equation}
		therefore we have
		\begin{equation}
			\begin{array}{lll}
				I_{31}+I_{32}+I_{33}
				& < & \displaystyle
				\frac{\eta}{\theta}+\frac{2 c_{1}}{x}+\frac{2 c_{1}}{x^{3}}+\frac{2 c_{2}}{y}+\frac{2 c_{2}}{y^{3}}\\[12pt]
				& & \displaystyle
				+\left(\frac{2 \eta}{\theta}+2\left(d+\gamma+\alpha_{1} C_{o}+\lambda_{1} C_{e}\right)
				+3 \sigma^{2} C_{e}^{2} +1\right) \frac{1}{x^{2}} \\[12pt]
				& & \displaystyle
				+\left(\frac{2 \eta}{\theta}+2\left(\alpha_{2} C_{o}+\frac{\lambda_{2} C_{e}^{p}}{1+\beta C_{e}^{q}}\right)
				+3 \sigma^{2} C_{e}^{2} +1\right) \frac{1}{y^{2}} \\[12pt]
				& & \displaystyle
				+\left(\frac{\eta}{\theta}+2\left(d+\gamma+\alpha_{1} C_{o}+\lambda_{1} C_{e}\right)
				+(3+2(\theta-1)) \sigma^{2} C_{e}^{2}\right) \frac{1}{x^{4}} \\[12pt]
				& & \displaystyle
				+\left(\frac{\eta}{\theta}+2\left(\alpha_{2} C_{o}+\frac{\lambda_{2} C_{e}^{p}}{1+\beta C_{e}^{q}}\right)
				+(3+2(\theta-1)) \sigma^{2} C_{e}^{2}\right) \frac{1}{y^{4}} \\[12pt]
				& & \displaystyle
				+2\left[\frac{\eta}{\theta}+\left(d+\gamma+\alpha_{1} C_{o}+\lambda_{1} C_{e}\right)
				+\left(\alpha_{2} C_{o}+\frac{\lambda_{2} C_{e}^{p}}{1+\beta C_{e}^{q}}\right)+3 \sigma^{2} C_{e}^{2}\right. \\[12pt]
				& & \displaystyle
				\left.+\frac{c_{1}^{2}+c_{2}^{2}}{2}-4 \sqrt{b \gamma}\right] \frac{1}{x^{2} y^{2}}-\frac{4 \sqrt{b \gamma}}{x
				y}.
			\end{array}
		\end{equation}
		By condition~\eqref{eq:stochastic_permanence} of Theorem~\ref{thm:stochastic_permanence}, so we get
		\begin{equation}
			\begin{array}{lll}
				I_{31}+I_{32}+I_{33}
				& < & \displaystyle
				\frac{\eta}{\theta}+\frac{2 c_{1}}{x}+\frac{2 c_{1}}{x^{3}}+\frac{2 c_{2}}{y}+\frac{2 c_{2}}{y^{3}}\\[12pt]
				& & \displaystyle
				+\left(\frac{2 \eta}{\theta}+2\left(d+\gamma+\alpha_{1}+\lambda_{1}\right)+3 \sigma^{2}+1\right) \frac{1}{x^{2}} \\[12pt]
				& & \displaystyle
				+\left(\frac{2 \eta}{\theta}+2\left(\alpha_{2}+\lambda_{2}\right)+3 \sigma^{2}+1\right) \frac{1}{y^{2}} \\[12pt]
				& & \displaystyle
				+\left(\frac{\eta}{\theta}+2\left(d+\gamma+\alpha_{1}+\lambda_{1}\right)+(3+2(\theta-1) \sigma^{2}\right) \frac{1}{x^{4}} \\[12pt]
				& & \displaystyle
				+\left(\frac{\eta}{\theta}+2\left(\alpha_{2}+\lambda_{2}\right)+(3+2(\theta-1) \sigma^{2}\right) \frac{1}{y^{4}} \\[12pt]
				& \coloneqq & F(x, y).
			\end{array}
		\end{equation}
		It is obvious that
		\begin{equation}
			H(x, y) \coloneqq F(x, y)G(x, y), \quad \, G(x, y)=\left(1+\frac{1}{x^{2}}+\frac{1}{y^{2}}\right)^{\theta-2}
		\end{equation}
		admits an upper boundary for all $x$ and $y$ in the domain
		$\mathbb{R}_{+}$.
		Let us denote
		\begin{equation}
			H_{1}=\sup _{(x, y) \in \mathbb{R}_{+}^{2}} H(x, y)<\infty.
		\end{equation}
		Integrating~\eqref{eq:dV3}, together with $x_{0}=x(0),~ y_{0}=y(0)$,
		which gives that
		\begin{equation}
			e^{\eta t}\left(1+\frac{1}{x^{2}}+\frac{1}{y^{2}}\right)^{\theta}
			<\left(1+\frac{1}{x_{0}^{2}}+\frac{1}{y_{0}^{2}}\right)^{\theta}
			+\frac{H_{1}\theta}{\eta}(e^{\eta t}-1)+\int_{0}^{t} M(s)\mathrm{d}s,
		\end{equation}
		taking the expectation on both sides, then
		\begin{equation}
			\label{eq:EV3}
			\mathbb{E}\left[\frac{1}{2^{\theta}}\left(1+\frac{1}{x^{2}}+\frac{1}{y^{2}}\right)^{\theta}\right]
			< \mathbb{E}\left[\frac{e^{-\eta\,t}}{2^{\theta}}\left(1+\frac{1}{x_{0}^{2}}+\frac{1}{y_{0}^{2}}\right)^{\theta}
			+\frac{H_{1}\theta}{2^{\theta}\eta}(1-e^{-\eta\,t})+\frac{e^{-\eta\,t}}{2^{\theta}}\int_{0}^{t}
			M(s)\mathrm{d}s\right].
		\end{equation}
		Therefore, together with~\eqref{eq:EV3}, taking supremum limit, we
		derive that
		\begin{equation}
			\label{eq:H2}
			\limsup _{t \rightarrow \infty} \mathbb{E}\left(\frac{1}{x+y}\right)^{2 \theta}
			\leqslant \limsup _{t \rightarrow \infty} \mathbb{E}\left(\frac{1}{x^{2}+y^{2}}\right)^{\theta}
			< \limsup _{t \rightarrow \infty}\left[\frac{1}{2^{\theta}}
			\left(1+\frac{1}{x^{2}}+\frac{1}{y^{2}}\right)^{\theta}\right]<\frac{H_{1}\theta}{2^{\theta}
				\eta}\coloneqq H_{2}.
		\end{equation}

		For any given $\varepsilon>0$, we let
		\begin{equation*}
			\xi=\left(\frac{\varepsilon}{H_{2}}\right)^{\frac{1}{2
			\theta}},
		\end{equation*}
		the Chebyshev inequality gives that
		\begin{equation*}
			\mathbb{P}\{x+y<\xi\}=\mathbb{P}\left\{(x+y)^{-\theta}>\xi^{-\theta}\right\}
			\leqslant \frac{\mathbb{E}\left((x+y)^{-2 \theta}\right)}{\xi^{-2 \theta}}=\xi^{2 \theta} \mathbb{E}\left((x+y)^{-2
			\theta}\right).
		\end{equation*}
		Combined with~\eqref{eq:H2}, we have
		\begin{equation}
			\lim _{t \rightarrow \infty}
			\mathbb{P}\{x+y<\xi\}<\varepsilon,
		\end{equation}
		which implies that
		\begin{equation}
			\liminf _{t \rightarrow \infty} \mathbb{P}\{x+y \geqslant \xi\} \geqslant
			1-\varepsilon.
		\end{equation}


		We define a $C^{2}$-function $
		V_{4}=e^{t}\left(x^{m}+y^{n}\right)^{l}, 0<l<1, m, n \in
		\mathbb{N}_{+}.$ By It\^{o}'s formula again
		\begin{equation*}
			\begin{array}{lll}
				\vspace{1ex}
				\mathrm{d} V_{4}
				& = & \displaystyle \left.e^{t}\left(x^{m}+y^{n}\right) \mathrm{d} t+l e^{t}\left(x^{m}+y^{n}\right)\right)^{l-1}\left(m x^{m-1} \mathrm{d} x+n y^{n-1} \mathrm{d} y\right)
				\\\vspace{1ex}
				& & \displaystyle +\frac{1}{2} l(l-1) e^{t}\left(x^{m}+y^{n}\right)^{l-2}\left(m^{2} x^{2(m-1)}(\mathrm{d} x)^{2}+n^{2} y^{2(n-1)}(\mathrm{d} y)^{2}\right)
				\\\vspace{1ex}
				& & \displaystyle +\frac{1}{2} l e^{t}\left(x^{m}+y^{n}\right)^{l-1}\left(m(m-1) x^{m-2}(\mathrm{d} x)^{2}+n(n-1) y^{n-2}(\mathrm{d} y)^{2}\right)
				\\\vspace{1ex}
				& = & \displaystyle \mathcal{L} V_{4} \mathrm{d} t -l e^{t}\left(x^{m}+y^{n}\right)^{l-1}\left(m \sigma_{1} x^{m} C_{e} \mathrm{d} B_{1}(t)
				+\frac{n \sigma_{2} y^{n} C_{e}^{p}}{1+\beta C_{e}^{q}} \mathrm{d}
				B_{2}(t)\right),
			\end{array}
		\end{equation*}
		where
		\begin{equation}
			\begin{array}{lll}
				\vspace{1ex}
				\mathcal{L} V_{4}
				& = & \displaystyle e^{t}(x^{m}+y^{n})^{l}
				+l e^{t}(x^{m}+y^{n})^{l-1}\Big[\Big(b \frac{y}{x}-d-\gamma-c_{1} x-\alpha_{1} C_{o}-\lambda_{1} C_{e}\Big) m x^{m}
				\\\vspace{1ex}
				& & \displaystyle + \Big(\gamma \frac{x}{y}-c_{2} y-\alpha_{2} C_{o}-\frac{\lambda_{1} C_{e}^{p}}{1+\beta C_{e}^{q}}\Big) n y^{n}\Big]
				\\\vspace{1ex}
				& & \displaystyle +\frac{1}{2} l(l-1) e^{t}(x^{m}+y^{n})^{l-2}\Big(m^{2} \sigma_{1}^{2} x^{2 m} C_{e}^{2}
				+\frac{n^{2} \sigma_{2}^{2} y^{2 n} C_{e}^{2 p}}{(1+\beta C_{e}^{q})^{2}}\Big)
				\\\vspace{1ex}
				& & \displaystyle +\frac{1}{2} l e^{t}(x^{m}+y^{n})^{l-1}\Big(m(m-1) \sigma_{1}^{2} x^{m} C_{e}^{2}
				+\frac{n(n-1) \sigma_{2}^{2} y^{n} C_{e}^{2 p}}{(1+\beta C_{e}^{q})^{2}}\Big)
				\\\vspace{1ex}
				& < & \displaystyle e^{t}(x^{m}+y^{n})^{l-1}\Big\{x^{m}+y^{n}
				+l\Big[\Big(b \frac{y}{x}-d-\gamma-c_{1} x-\alpha_{1} C_{o}-\lambda_{1} C_{e}\Big) m x^{m}
				\\\vspace{1ex}
				& & \displaystyle +\Big(\gamma \frac{x}{y}-c_{2} y-\alpha_{2} C_{o}-\frac{\lambda_{1} C_{e}^{p}}{1+\beta C_{e}^{q}}\Big) n y^{n}\Big]
				\\\vspace{1ex}
				& & \displaystyle +\frac{1}{2} l\Big(m(m-1) \sigma_{1}^{2} x^{m} C_{e}^{2}
				+\frac{n(n-1) \sigma_{2}^{2} y^{n} C_{e}^{2 p}}{(1+\beta C_{e}^{q})^{2}}\Big)\Big\}
				\\\vspace{1ex}
				& \leqslant & \displaystyle e^{t}(x^{m}+y^{n})^{l-1}\Big\{x^{m}+y^{n}
				\\\vspace{1ex}
				& & \displaystyle +\frac{b m l}{2}(x^{2(m-1)}+y^{2})+\frac{m l}{2}(x^{2}+y^{2(n-1)})-c_{1} m l x^{m+1}-c_{2} n l y^{n+1}
				\\\vspace{1ex}
				& & \displaystyle +l\Big[\Big(-d-\alpha_{1} C_{o}-\lambda_{1} C_{e}\Big) m x^{m}
				+\Big(-\alpha_{2} C_{o}-\frac{\lambda_{2} C_{e}^{p}}{1+\beta C_{e}^{q}}\Big) n y^{n}\Big]
				\\\vspace{1ex}
				& & \displaystyle +\frac{1}{2} l\Big(m(m-1) \sigma_{1}^{2} x^{m} C_{e}^{2}
				+\frac{n(n-1) \sigma_{2}^{2} y^{n} C_{e}^{2 p}}{(1+\beta C_{e}^{q})^{2}}\Big)\Big\}
				\\\vspace{1ex}
				& = & \displaystyle e^{t}(x^{m}+y^{n})^{l-1}\Big\{-c_{1} m l x^{m+1}-c_{2} n l y^{n+1}
				\\\vspace{1ex}
				& & \displaystyle +\Big[1+\frac{1}{2} m(m-1) l \sigma_{1}^{2} C_{e}^{2}-m l(d+\alpha_{1} C_{o}+\lambda_{1} C_{e})\Big] x^{m}
				\\\vspace{1ex}
				& & \displaystyle +\Big[1+\frac{n(n-1) l \sigma_{2}^{2} C_{e}^{2 p}}{2(1+\beta C_{e}^{q})^{2}}
				-n l\Big(\alpha_{2} C_{o}+\frac{\lambda_{2} C_{e}^{p}}{1+\beta C_{e}^{q}}\Big)\Big] y^{n}
				\\\vspace{1ex}
				& & \displaystyle +\frac{b m l}{2}(x^{2(m-1)}+y^{2})+\frac{m
				l}{2}(x^{2}+y^{2(n-1)})\Big\}.
			\end{array}
		\end{equation}

		We take $m=n=2$, then
		\begin{equation}
			\label{eq:LV4m=n=2}
			\begin{array}{lll}
				\vspace{1ex}
				\mathcal{L} V_{4}
				& < & e^{t}\left(x^{2}+y^{2}\right)^{l-1}\left\{-2 c_{1} l x^{3}-2 c_{2} l y^{3}\right.
				\\\vspace{1ex}
				& & +\left[1+l \sigma_{1}^{2} C_{e}^{2}-2 l\left(d+\alpha_{1} C_{o}+\lambda_{1} C_{e}\right)+b l+\gamma\right] x^{2}
				\\\vspace{1ex}
				& & \displaystyle \left.+\left[1+\frac{l \sigma_{2}^{2} C_{e}^{2 p}}{\left(1+\beta C_{e}^{q}\right)^{2}}-2 l\left(\alpha_{2} C_{o}+\frac{\lambda_{2} C_{e}^{p}}{1+\beta C_{e}^{q}}\right)+b l+\gamma\right] y^{2}\right\}
				\\\vspace{1ex}
				& \coloneqq & e^{t} J(x, y)  \leqslant C e^{t}
			\end{array}
		\end{equation}
		taking the expectation on both sides of~\eqref{eq:LV4m=n=2}, then
		\begin{equation}
			\mathbb{E}(e^{t}(x^{2}+y^{2})^{l})-(x^{2}(0)+y^{2}(0))^{l}<C(e^{t}-1),
		\end{equation}
		where $ C $ is a constant.
		Due to $ x+y \leqslant x^{2}+y^{2},$ we have
		\begin{equation}
			\label{eq:3.49}
			\liminf _{t \rightarrow \infty}
			\mathbb{E}\left(x^{2}+y^{2}\right)<C.
		\end{equation}
		For any given $\varepsilon>0$, we let
		\begin{equation}
			N=\left(\frac{C}{\varepsilon}\right)^{\frac{\,1}{\,l}},
		\end{equation}
		by Chebyshev inequality
		\begin{equation}
			\mathbb{P}\{x+y>N\}=\mathbb{P}\{(x+y)^{\frac{l}{2}}>N^{\frac{l}{2}}\} \leqslant
			\frac{\mathbb{E}(x+y)^{l}}{N^{l}}.
		\end{equation}
		Combined with~\eqref{eq:3.49}, we have
		\begin{equation}
			\liminf _{t \rightarrow \infty}
			\mathbb{P}\{x+y>N\}<\varepsilon,
		\end{equation}
		which implies that
		\begin{equation}
			\limsup _{t \rightarrow \infty} \mathbb{P}\{x+y \leqslant N\} \geqslant
			1-\varepsilon.
		\end{equation}

		In other words, the single population is random and persistent.
	\end{proof}

	\section{Numerical simulations}
	\label{sec:simulations}

	To support the main conclusions of this paper, we
	selected the appropriate parameters and initial values for each
	conclusion.
	The main method of numerical simulation is to
	discretize the original system and use matlab to observe the trend
	of the graph.
	Then verify the conclusions of this paper.
	At the
	same time of verifying random and persistent survival, under the
	premise that the remaining parameters are fixed, we take 4 sets of
	different psychological effect intensities for numerical simulation.

	The parameters and their reference value are shown on this table:
	\begin{table}[H]
		\begin{center}
			\begin{minipage}{\textwidth}
				\caption{Parameters and their reference value}
				\label{tab:parameters_value}
				\begin{tabular*}{\textwidth}{@{\extracolsep{\fill}}cccl@{\extracolsep{\fill}}}
					\toprule
					Parameter         & Value     & Reference value & Source                           \\
					\midrule
					$b$                    & 0.3       & [0.1,0.3]       & ~\cite{RN34,RN37,RN41}           \\
					$d$                    & 0.01      & 0.01            & ~\cite{RN34}                     \\
					$\gamma$               & 0.5       & [0.2,0.4]       & ~\cite{RN37,RN41}                \\
					$c_1$                  & 0.02      & -               & Assumed                          \\
					$c_2$                  & 0.05      & [0.1,0.4]       & ~\cite{RN34,RN37,RN41}           \\
					$\alpha_1$             & 0.2       & -               & Assumed                          \\
					$\alpha_2$             & 0.1       & [0.1,0.5]       & ~\cite{RN34,RN37,RN41}           \\
					$\lambda_1$            & 0.07      & -               & Assumed                          \\
					$\lambda_2$            & 0.05      & 0.05            & ~\cite{RN34,RN37}                \\
					$k$                    & 0.1       & [0.1,0.6]       & ~\cite{RN34,RN57,RN37,RN41}      \\
					$g$                    & 0.08      & [0.08,0.3]      & ~\cite{RN34,RN57,RN37,RN41}      \\
					$m$                    & 0.04      & [0.04,0.3]      & ~\cite{RN34,RN57,RN37,RN41}      \\
					$h$                    & 0.8       & [0.1,0.5]       & ~\cite{RN34,RN57,RN37,RN58,RN41} \\
					$u_e$                  & [0,0.6]   & [0,0.8]         & ~\cite{RN34,RN57,RN41}           \\
					$\beta$                & [0.1,10]  & [0.1,10]        & ~\cite{RN34,RN37}                \\
					$\sigma_1$, $\sigma_2$ & [0.1,0.3] & [0.1,0.3]       & ~\cite{RN34}                     \\
					\bottomrule
				\end{tabular*}
			\end{minipage}
		\end{center}
	\end{table}

	\subsection{Simulation of globally unique positive solution}
	\label{subsec:simulations_unique_positive_solution}

	The solution of equations~\eqref{eq:RewrittenEquationGroup} in the first quadrant is the globally unique positive solution, noted as $ (A, B) $.
	Take the value in Table~\ref{tab:parameters_value}, then we can get $A = 2.8078$ and $ B = 5.2989$.
	We draw and calibrate $(A, B)$ as shown in the Figure~\ref{fig:AB}.

	\begin{figure}[H]
		\centering
		\includegraphics[width=0.8\textwidth]{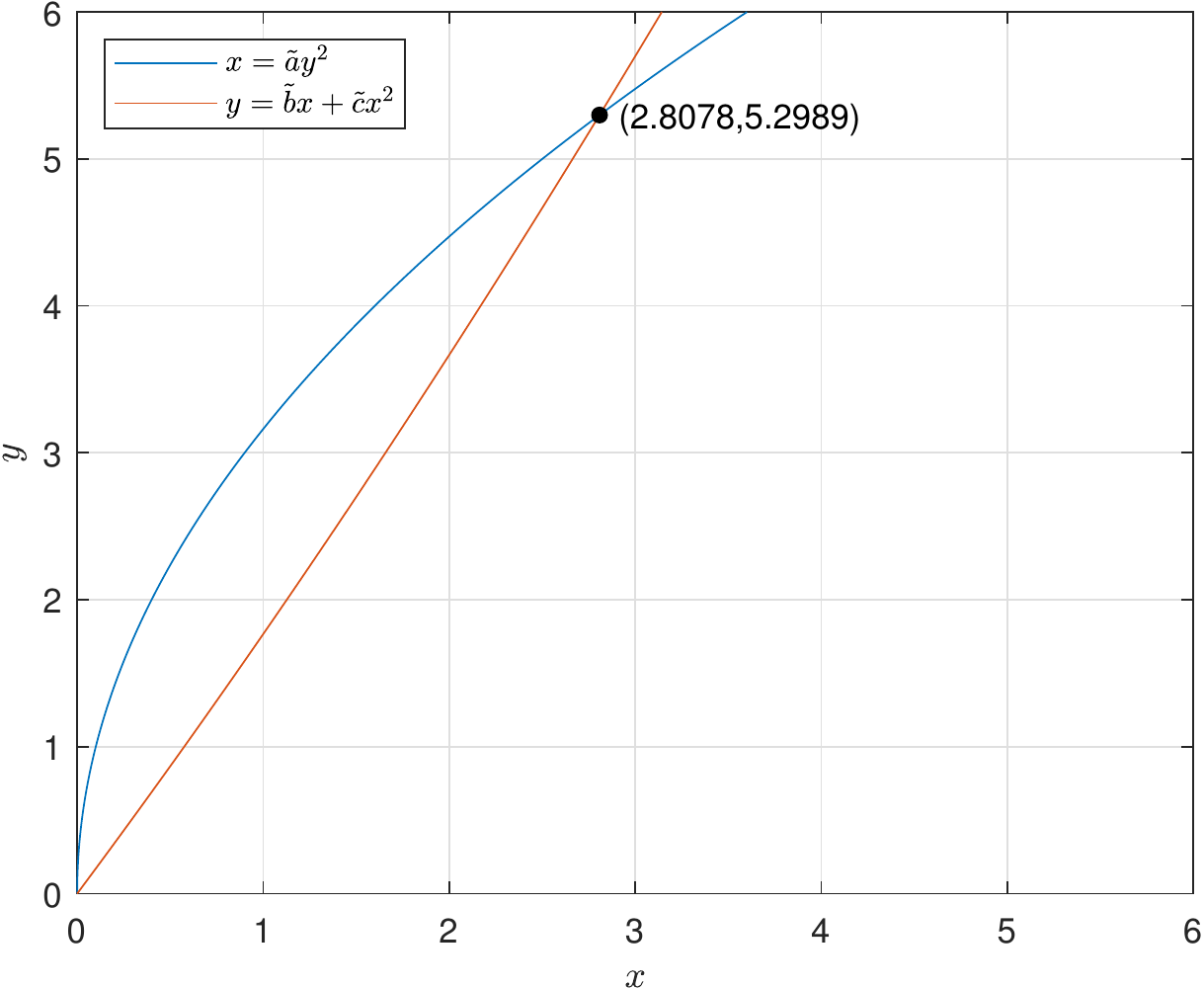}
		\caption{The globally unique positive solution of equations~\eqref{eq:RewrittenEquationGroup}}
		\label{fig:AB}
	\end{figure}

	\subsection{Simulations of weakly persistent in the mean}
	\label{subsec:simulations_weakly_persistent}

	Here we use numerical simulation results to support the conclusion of Theorem 3.
	We determine the parameters and initial values that meet the conditions of the theorem.
	Respectively, we select two toxin input levels for numerical simulation, through the image branch of numerical simulation to support the conclusion of subsection~\ref{subsec:weakly_persistent}.

	Considering that theorem 3 requires non-pollution, according to the actual situation, we take the smaller toxin input level $u_e\left(t\right)=0.1$ and $u_e\left(t\right)=0.1+0.1 \sin(t)$, and the psychological effect intensity $\beta= 0.1$, and then take the remaining parameters the same as Table~\ref{tab:parameters_value} while the variable parameters are shown on this Table~\ref{tab:variable_Th3}:
	\begin{table}[H]
		\begin{center}
			\begin{minipage}{\textwidth}
				\caption{Variable parameters and their value}
				\label{tab:variable_Th3}
				\begin{tabular*}{\textwidth}{@{\extracolsep{\fill}}ccccc@{\extracolsep{\fill}}}
					\toprule
					Group & $u_e$             & $\beta$ & $\sigma_1$ & $\sigma_2$ \\
					\midrule
					Group 1    & 0.1               & 0.1     & 0.15       & 0.1        \\
					Group 2    & $0.1+0.1 \sin(t)$ & 0.1     & 0.15       & 0.1        \\
					\bottomrule
				\end{tabular*}
			\end{minipage}
		\end{center}
	\end{table}

	The diagrams~\ref{fig:Th3} below track the density curves of $ x(t) $ and $ y(t)
	$ in Figure~\ref{fig:Th3_1} and Figure~\ref{fig:Th3_2} and toxin concentration $ C_o(t) $ and $ C_e(t) $ over time in Figure~\ref{fig:Th3_3} and Figure~\ref{fig:Th3_4}, when this population is weakly persistent in the mean.

	\begin{figure}[H]
		\centering
		\begin{subfigure}{0.45\textwidth}
			\includegraphics[width=\textwidth]{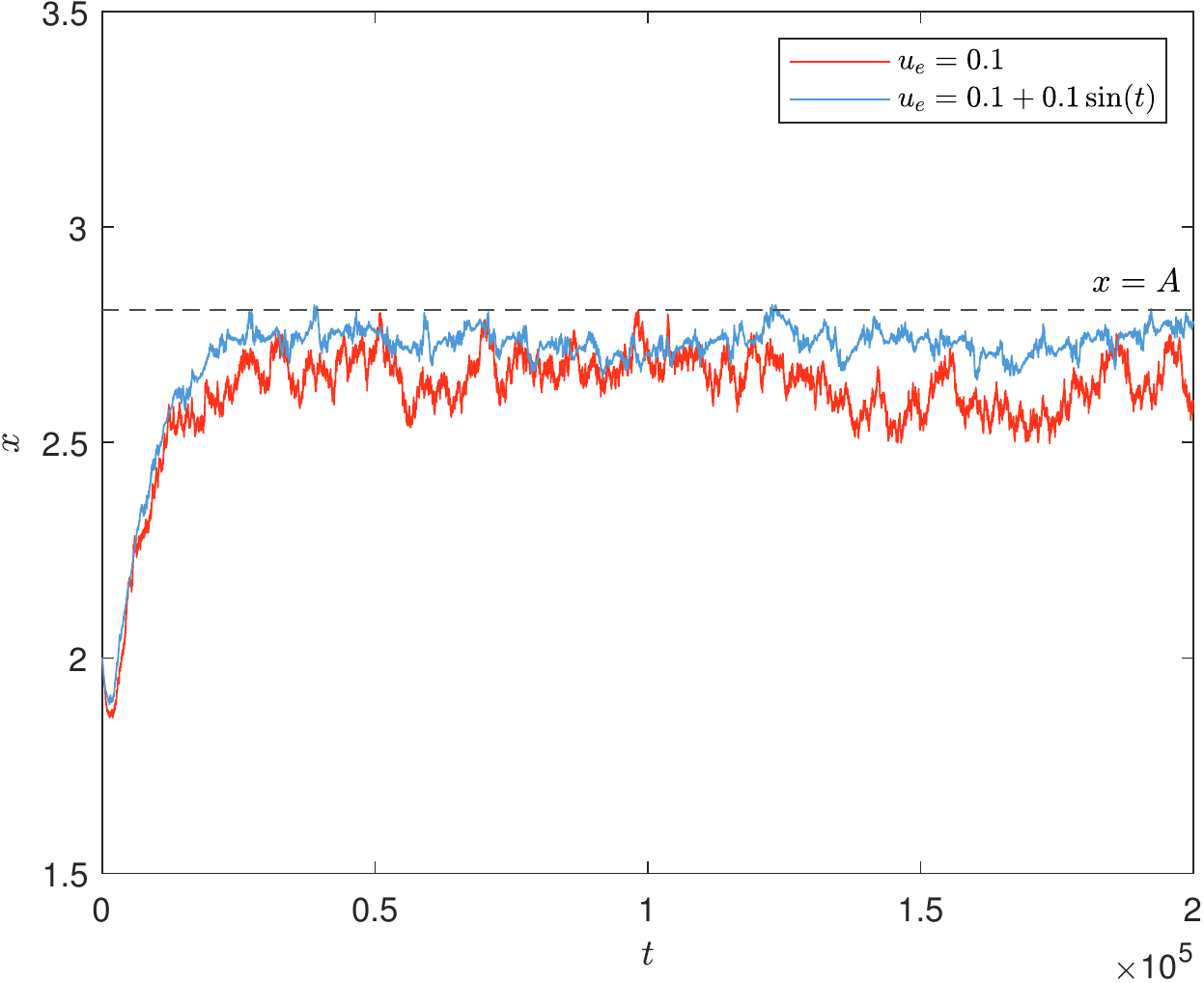}
			\caption{The density curves of $ x(t) $}
			\label{fig:Th3_1}
		\end{subfigure}
		\hfill
		\begin{subfigure}{0.45\textwidth}
			\includegraphics[width=\textwidth]{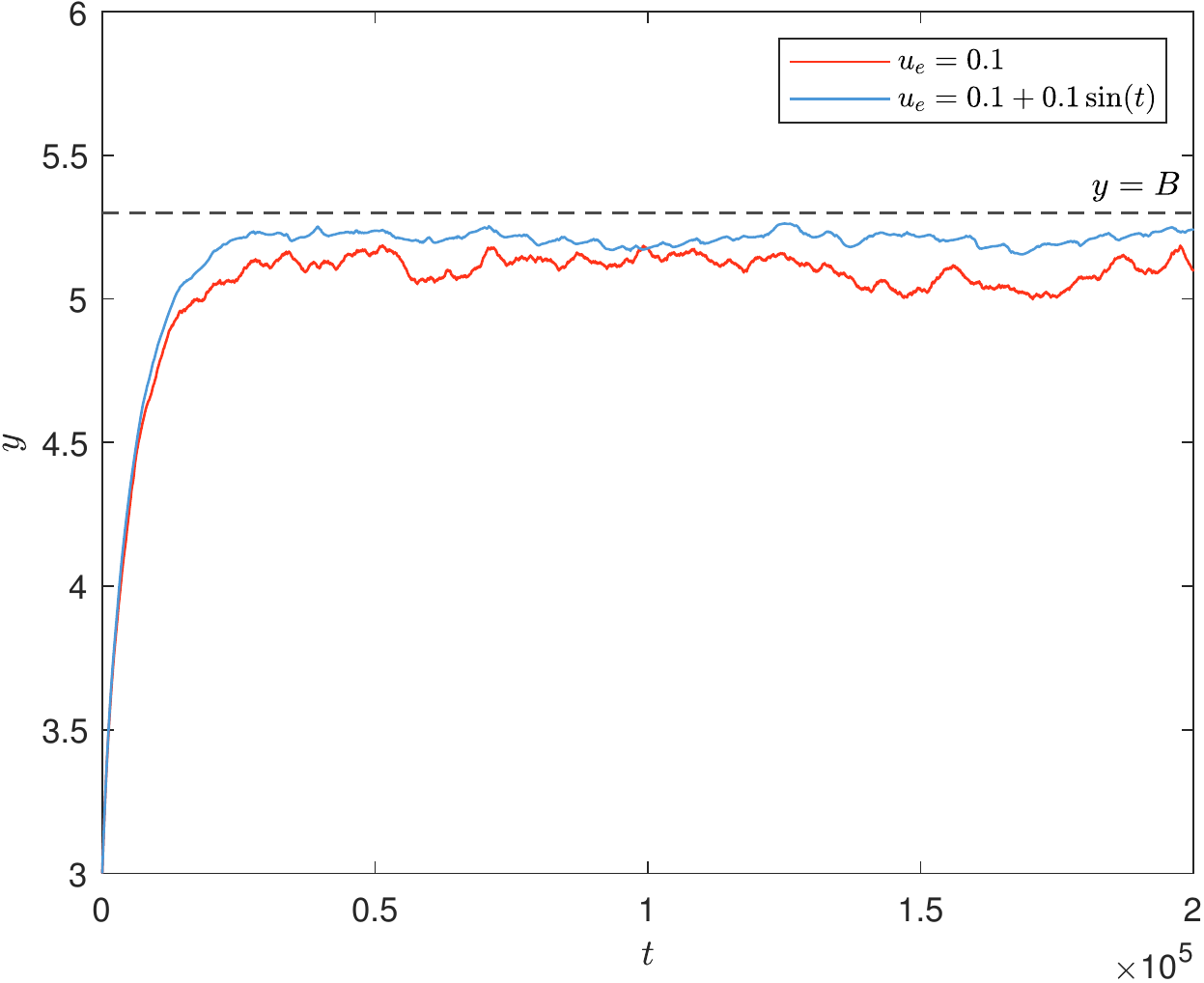}
			\caption{The density curves of $ y(t) $}
			\label{fig:Th3_2}
		\end{subfigure}
		\hfill
		\begin{subfigure}{0.45\textwidth}
			\includegraphics[width=\textwidth]{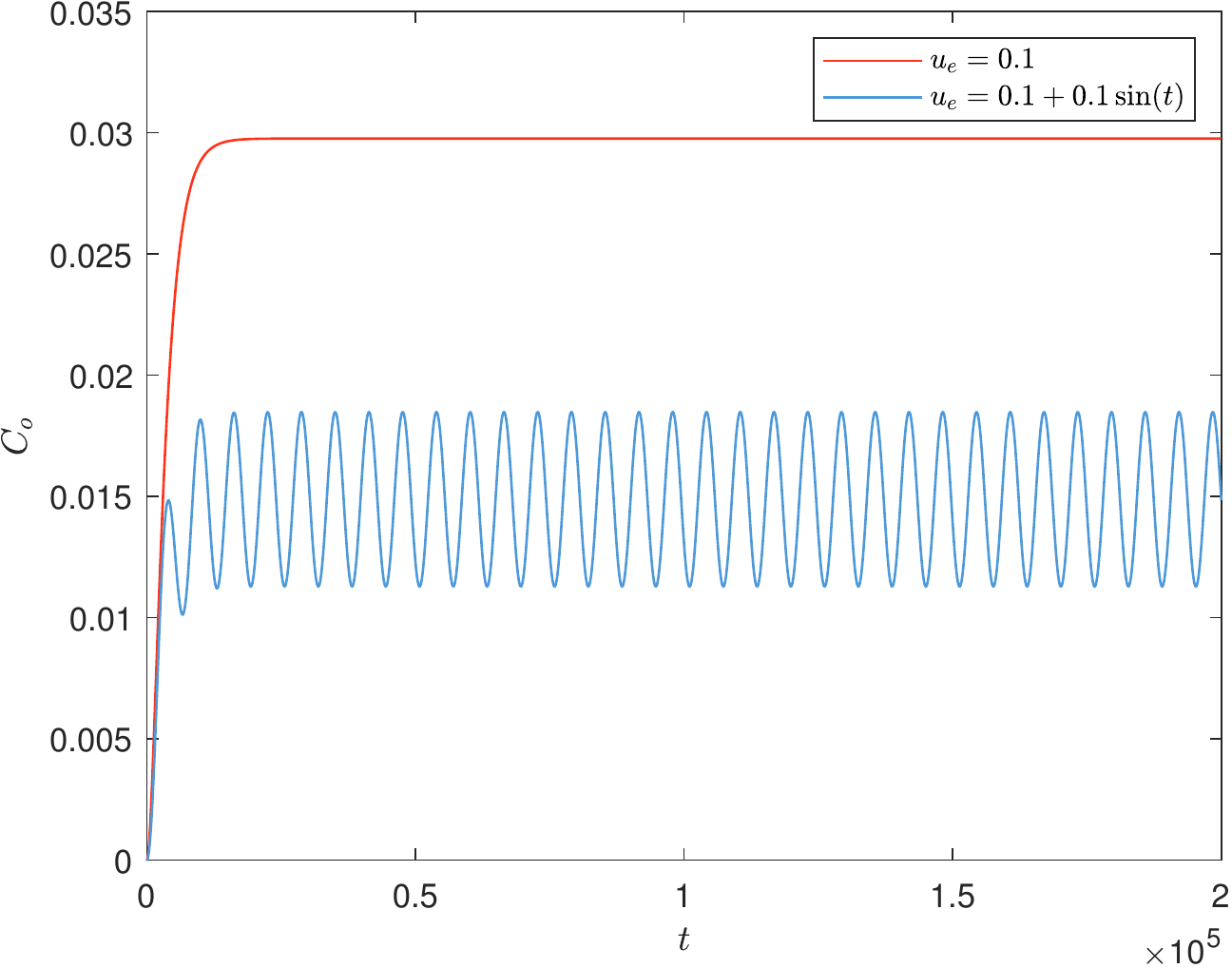}
			\caption{Toxin concentration $ C_o(t) $}
			\label{fig:Th3_3}
		\end{subfigure}
		\hfill
		\begin{subfigure}{0.45\textwidth}
			\includegraphics[width=\textwidth]{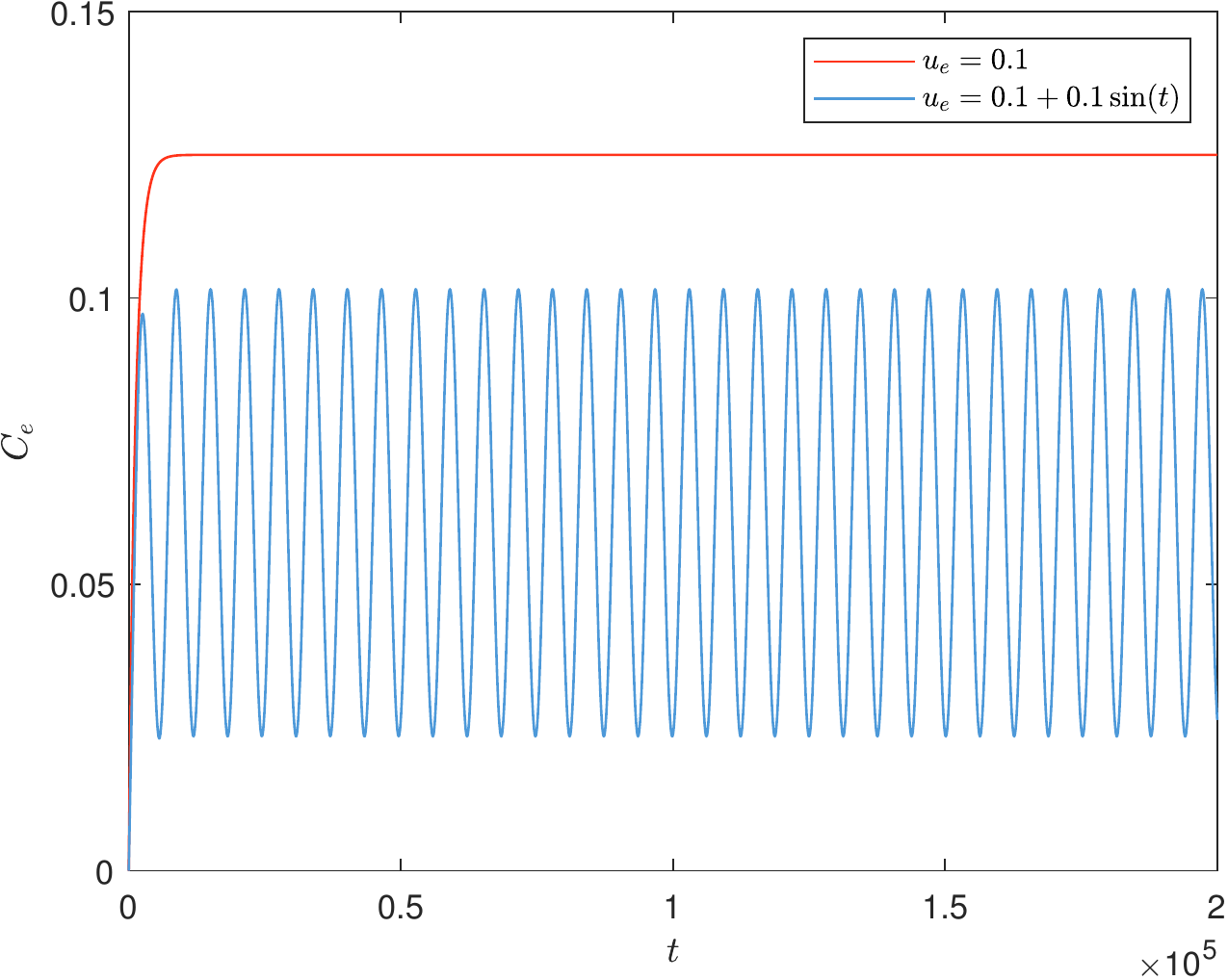}
			\caption{Toxin concentration $ C_e(t) $}
			\label{fig:Th3_4}
		\end{subfigure}
		\caption{The density curves of $ x(t) $ and $ y(t)$, Toxin concentration $ C_o(t) $ and $ C_e(t) $ with different $ u_e $.}
		\label{fig:Th3}
	\end{figure}

	\subsection{Simulations of stochastic permanence}
	\label{subsec:simulations_stochastic_permanence}

	We take the parameters and initial values that meet the conditions of the Theorem~\ref{thm:stochastic_permanence}, select the input levels of two toxins respectively, and group them according to the different noise intensity.
	At the same time, in order to verify the influence of the intensity of psychological effects on the dynamics of a single population, we divide them under the premise that the other parameters are fixed.
	We determine the variable parameters as in Table~\ref{tab:variable_Th4}, while others are consistent with Table~\ref{tab:parameters_value}.

	The variable parameters are shown on this table:
	\begin{table}[H]
		\begin{center}
			\begin{minipage}{\textwidth}
				\caption{Variable parameters and their value}
				\label{tab:variable_Th4}
				\begin{tabular*}{\textwidth}{@{\extracolsep{\fill}}ccccc@{\extracolsep{\fill}}}
					\toprule
					Figure & $u_e$             & $\beta$ & $ \sigma_1, \sigma_2$ \\
					\midrule
					a           & 0.3               & 0.1     & 0.1 or 0.3            \\
					b           & 0.3               & 1       & 0.1 or 0.3            \\
					c           & 0.3               & 10      & 0.1 or 0.3            \\
					d           & $0.3+0.3 \sin(t)$ & 0.1     & 0.1 or 0.3            \\
					e           & $0.3+0.3 \sin(t)$ & 1       & 0.1 or 0.3            \\
					f           & $0.3+0.3 \sin(t)$ & 10      & 0.1 or 0.3            \\
					\bottomrule
				\end{tabular*}
			\end{minipage}
		\end{center}
	\end{table}

	The three diagrams in Figure~\ref{fig:Th4x}, Figure~\ref{fig:Th4y} and Figure~\ref{fig:Th4CeCo} respectively track the density curves of $x(t)$ and $ y(t) $, toxin concentration $ C_o(t) $ and $ C_e(t) $    over time, when this population is stochastic permanence.

	We first take $u_e(t)= 0.3$, and on the premise of fixing the noise intensity of the two groups, respectively take four groups of psychological effect intensity $\beta= 0.1, 1, 10$ for comparison.
	The results are shown in Figure~\ref{fig:Th4x_1}-\ref{fig:Th4x_3} and Figure~\ref{fig:Th4y_1}-\ref{fig:Th4y_3}.

	Then we take $u_e\left(t\right)=0.3+0.3 \sin(t)$, on the premise of fixing the noise intensity of the two groups, respectively take four groups of psychological effect intensity $\beta= 0.1, 1, 10$ for comparison, the results are shown in Figure~\ref{fig:Th4x_4}-\ref{fig:Th4x_6} and Figure~\ref{fig:Th4y_4}-\ref{fig:Th4y_6}:

	The results show that when the psychological effect
	intensity increases, the adult population size will tend to be
	stable, which means that adults will actively evade the living
	area with high toxin concentration.

	\begin{figure}[H]
		\centering
		\begin{subfigure}[b]{0.32\textwidth}
			\includegraphics[width=\textwidth]{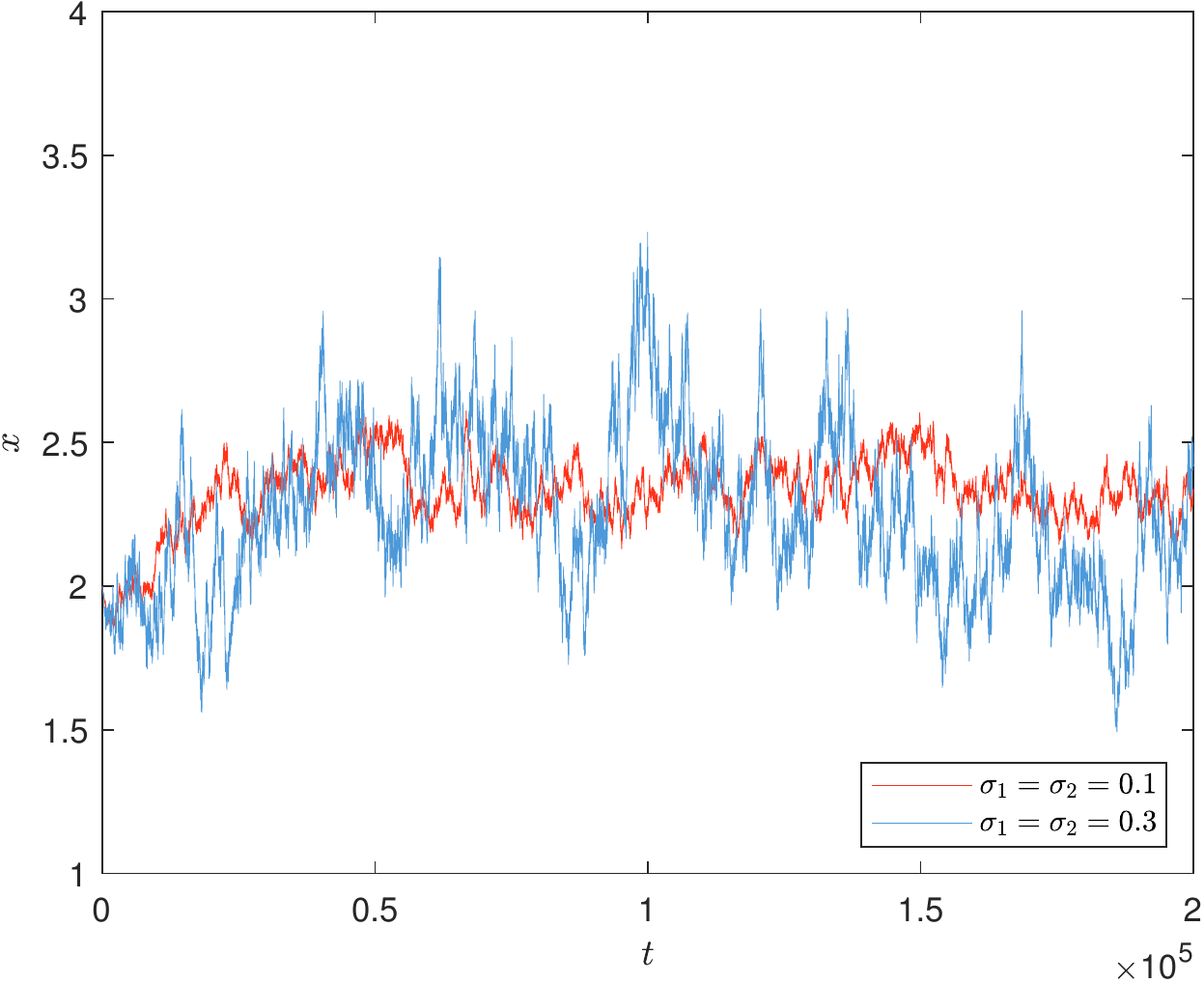}
			\caption{$ u_e=0.3, \beta=0.1 $}
			\label{fig:Th4x_1}
		\end{subfigure}
		\hfill
		\begin{subfigure}[b]{0.32\textwidth}
			\includegraphics[width=\textwidth]{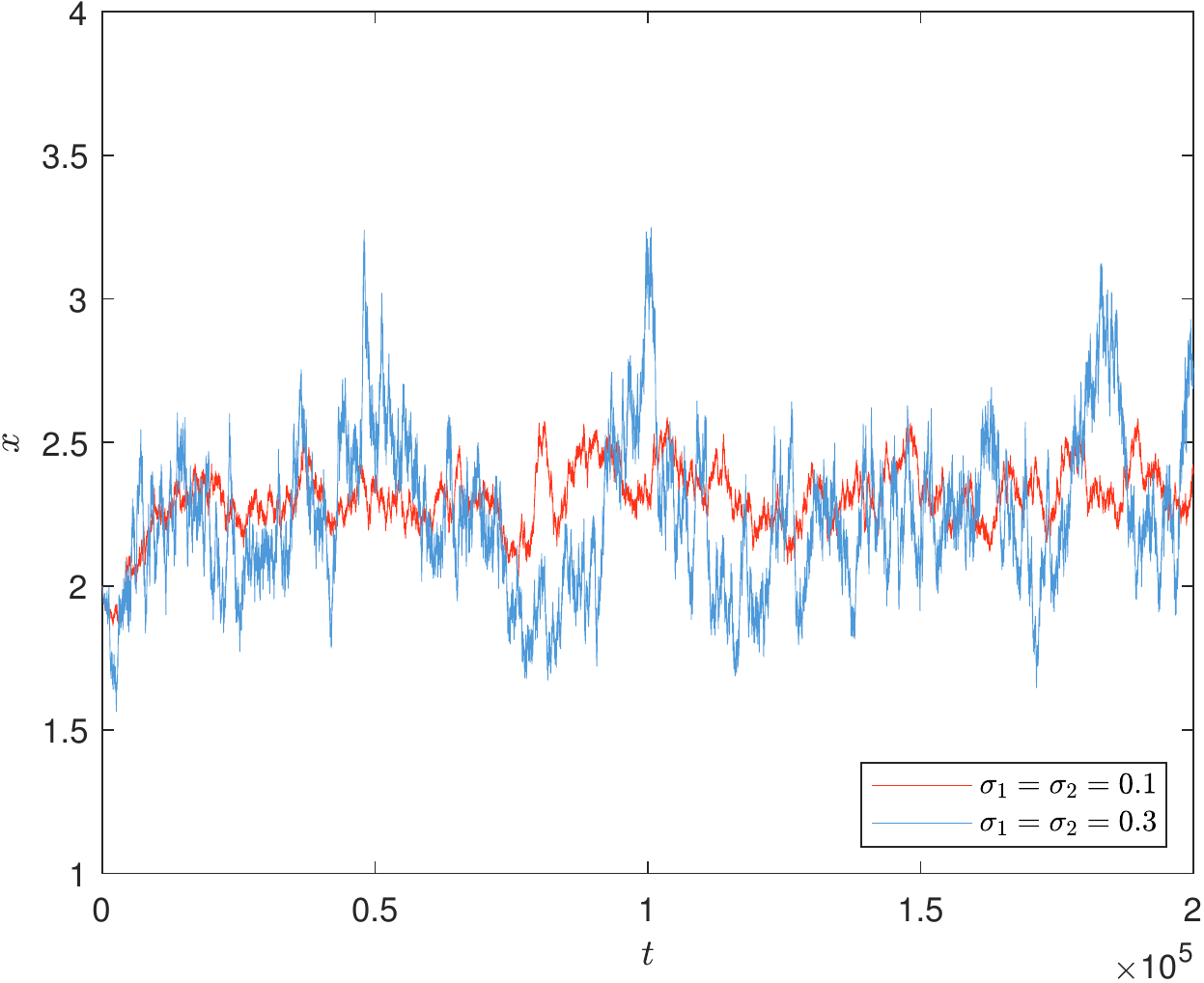}
			\caption{$ u_e=0.3, \beta=1 $}
			\label{fig:Th4x_2}
		\end{subfigure}
		\hfill
		\begin{subfigure}[b]{0.32\textwidth}
			\includegraphics[width=\textwidth]{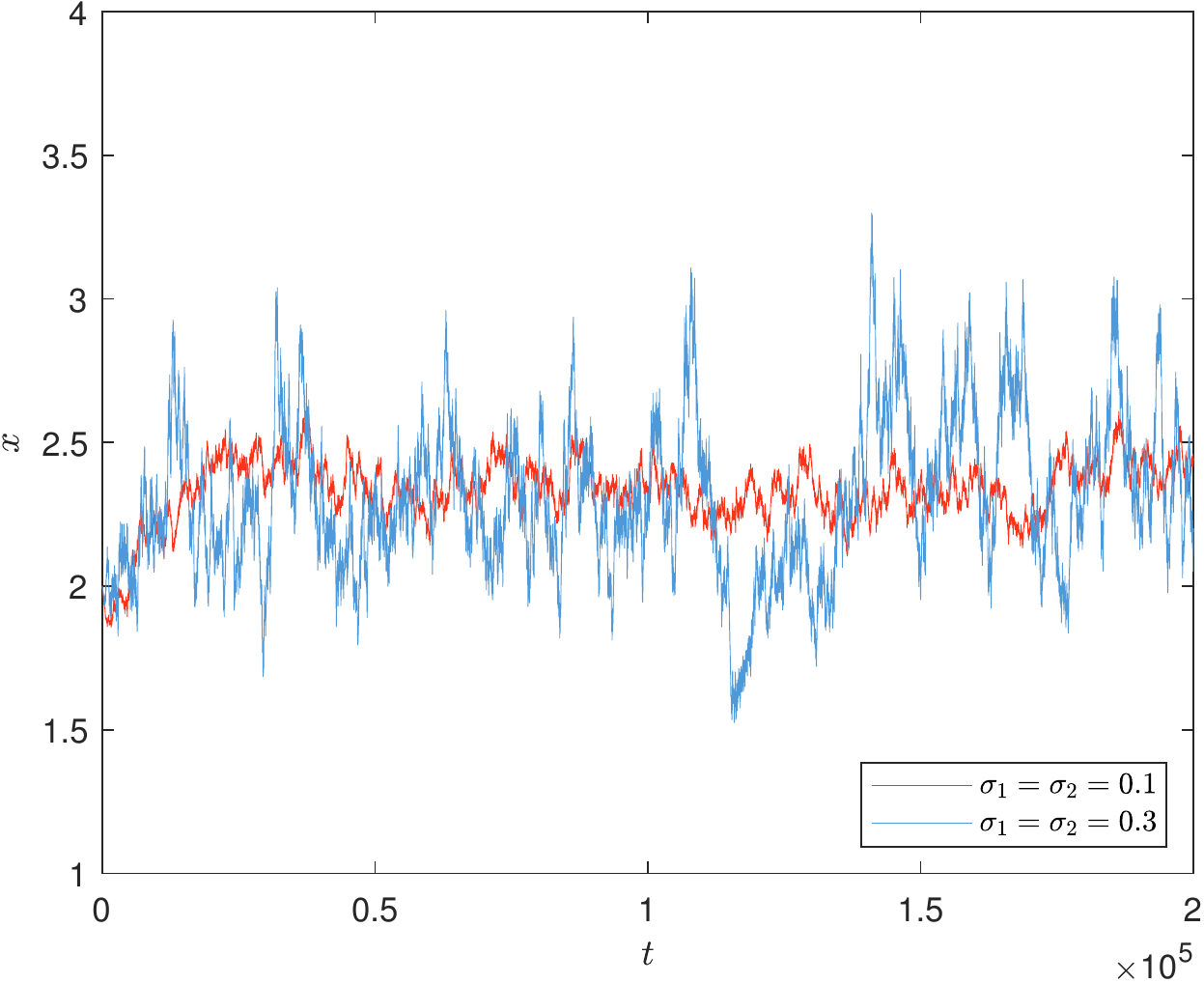}
			\caption{$ u_e=0.3, \beta=10 $}
			\label{fig:Th4x_3}
		\end{subfigure}
		\hfill
		\begin{subfigure}[b]{0.32\textwidth}
			\includegraphics[width=\textwidth]{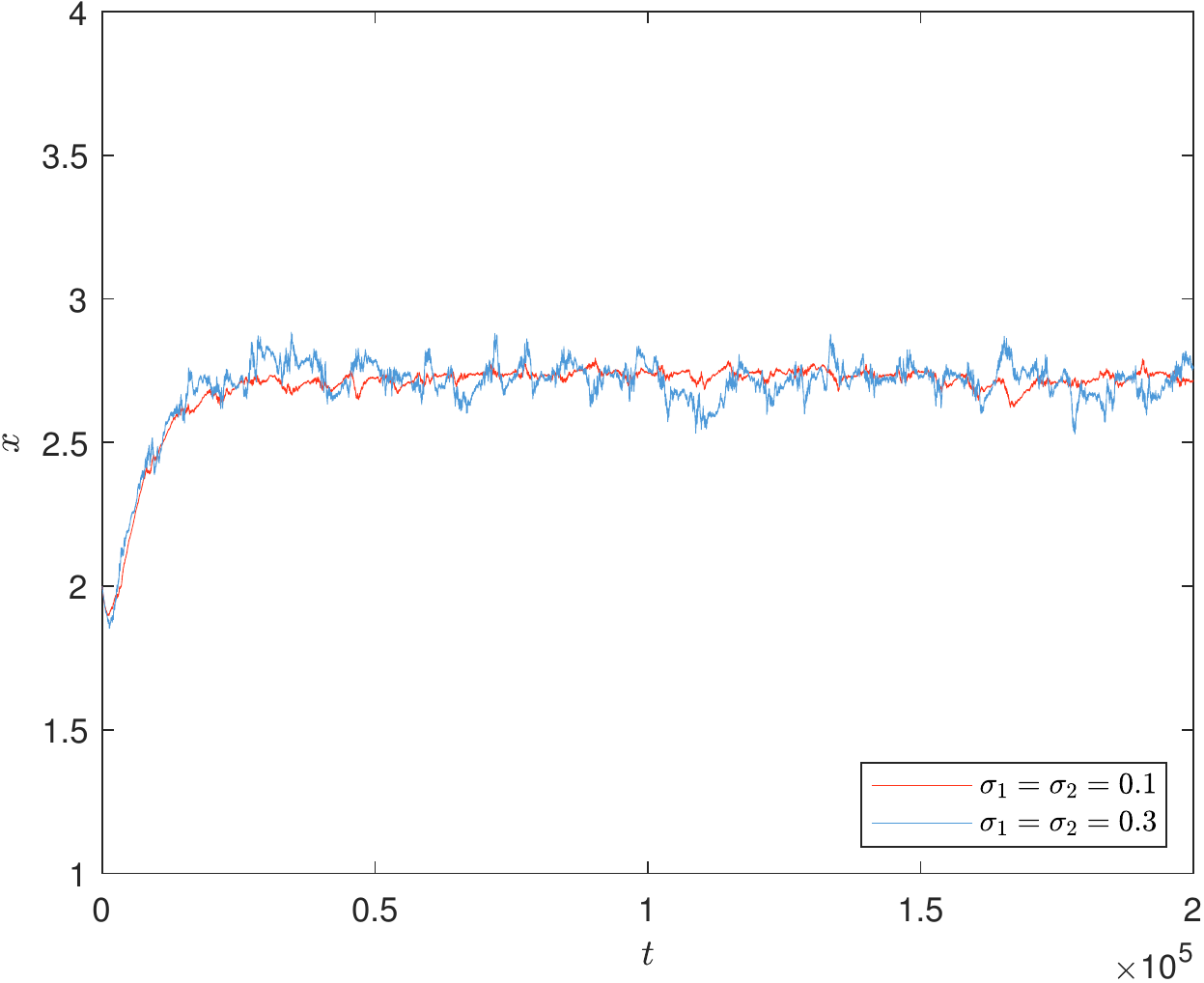}
			\caption{$ u_e=0.3+0.3\sin(t), \beta=0.1 $}
			\label{fig:Th4x_4}
		\end{subfigure}
		\hfill
		\begin{subfigure}[b]{0.32\textwidth}
			\includegraphics[width=\textwidth]{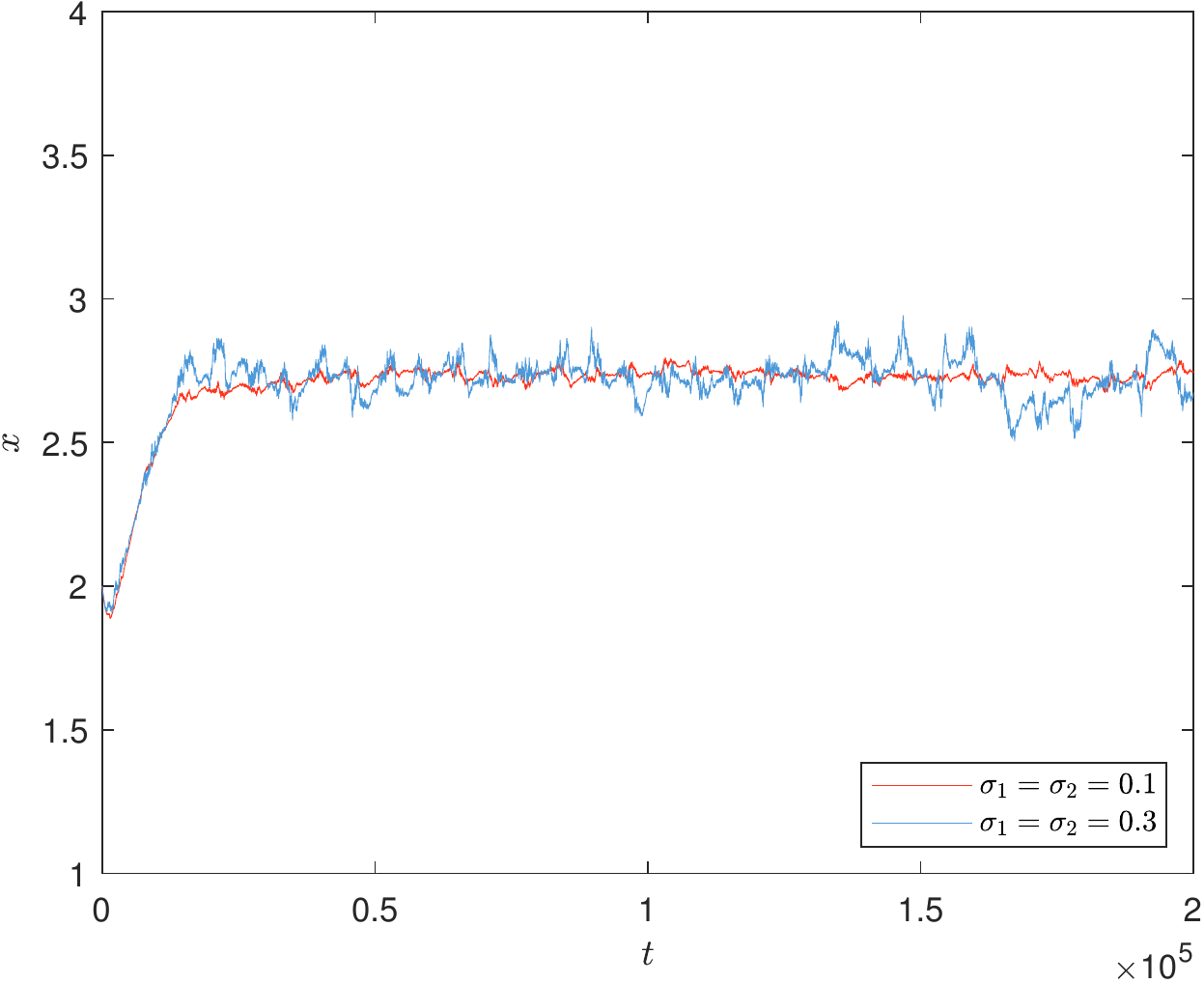}
			\caption{$ u_e=0.3+0.3\sin(t), \beta=1 $}
			\label{fig:Th4x_5}
		\end{subfigure}
		\hfill
		\begin{subfigure}[b]{0.32\textwidth}
			\includegraphics[width=\textwidth]{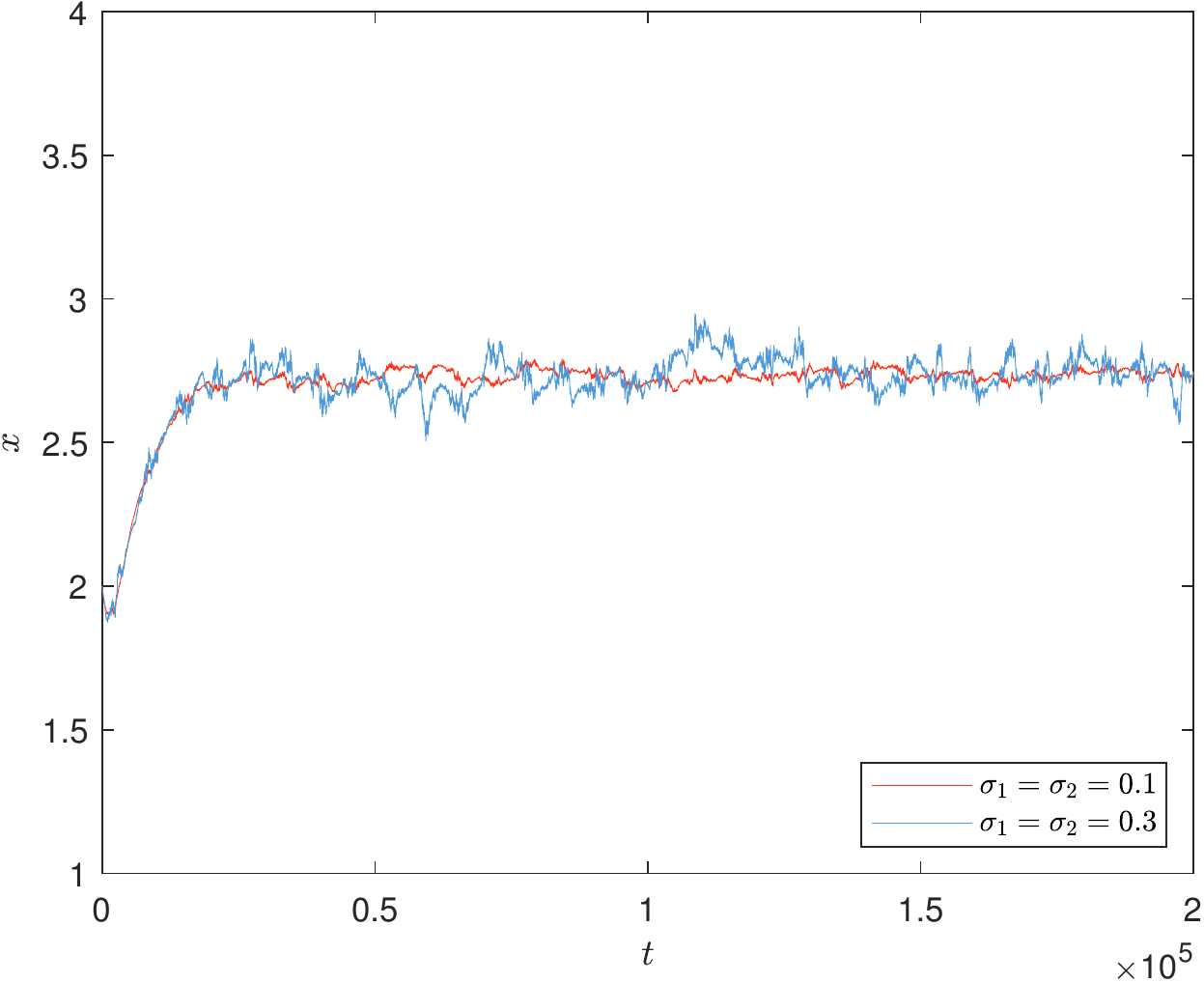}
			\caption{$ u_e=0.3+0.3\sin(t), \beta=10 $}
			\label{fig:Th4x_6}
		\end{subfigure}
		\caption{The density curves of $ x(t) $ with different $ u_e $ and $ \beta $.}
		\label{fig:Th4x}
	\end{figure}

	\begin{figure}[H]
		\centering
		\begin{subfigure}[b]{0.32\textwidth}
			\includegraphics[width=\textwidth]{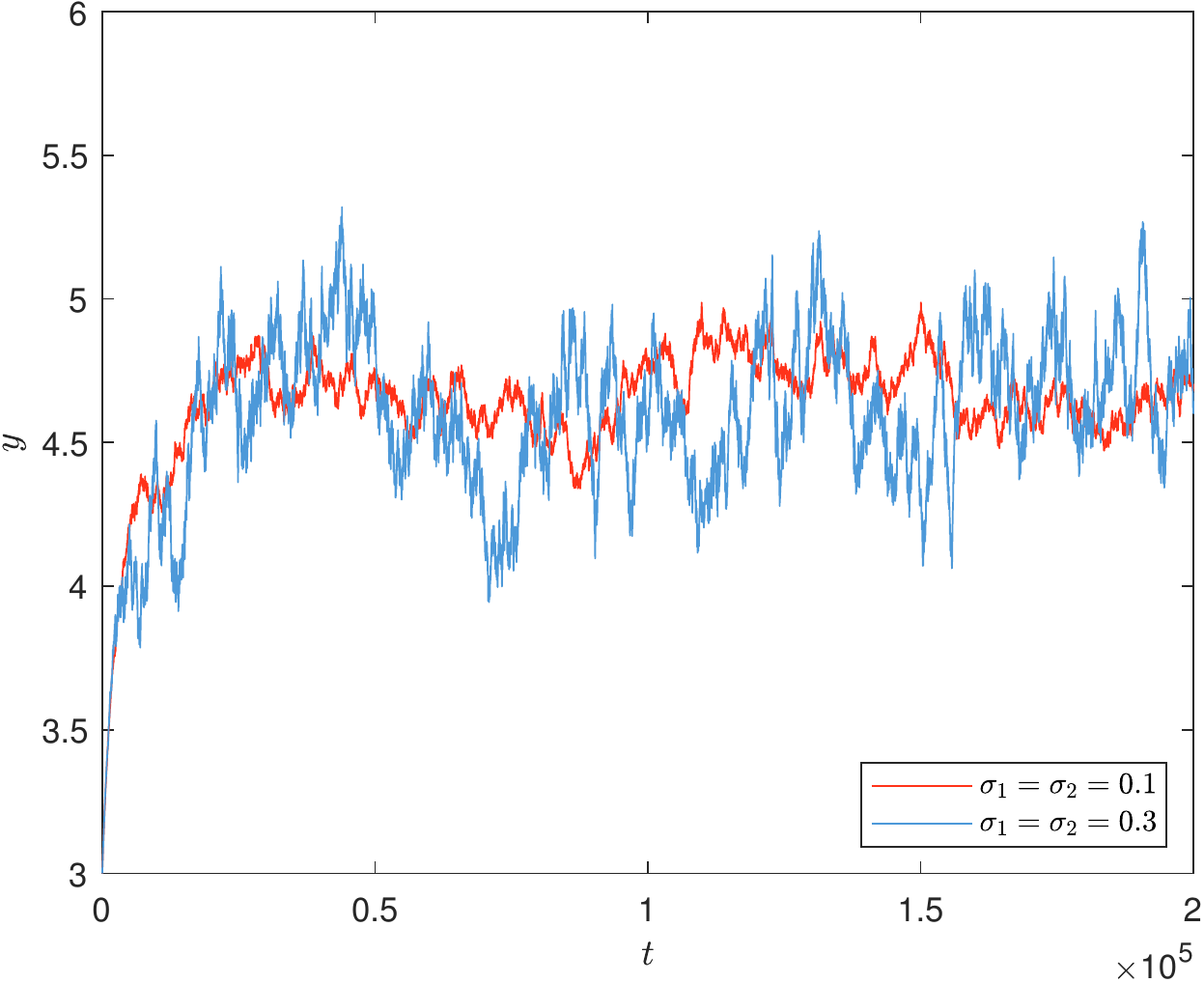}
			\caption{$ u_e=0.3, \beta=0.1 $}
			\label{fig:Th4y_1}
		\end{subfigure}
		\hfill
		\begin{subfigure}[b]{0.32\textwidth}
			\includegraphics[width=\textwidth]{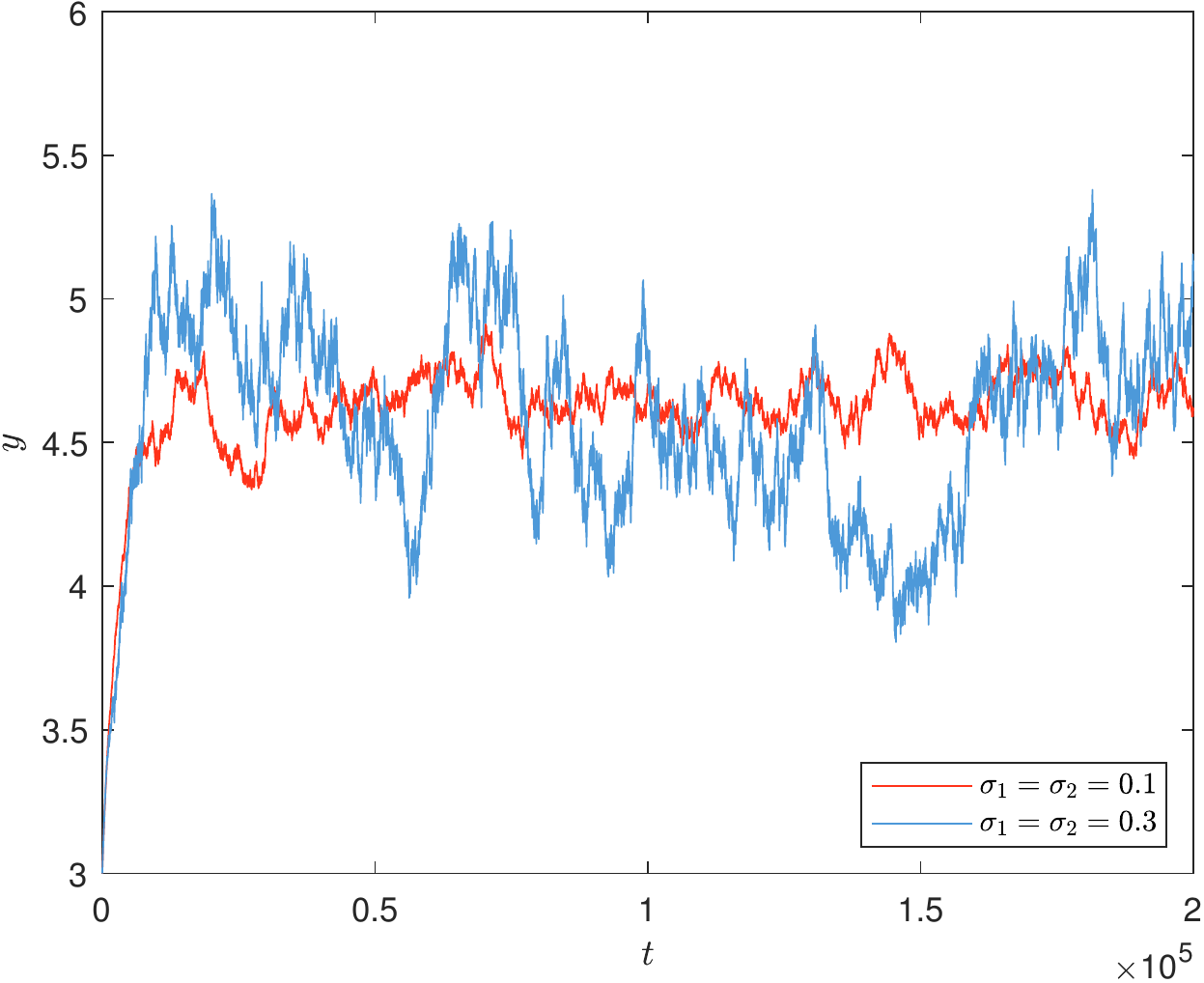}
			\caption{$ u_e=0.3, \beta=1 $}
			\label{fig:Th4y_2}
		\end{subfigure}
		\hfill
		\begin{subfigure}[b]{0.32\textwidth}
			\includegraphics[width=\textwidth]{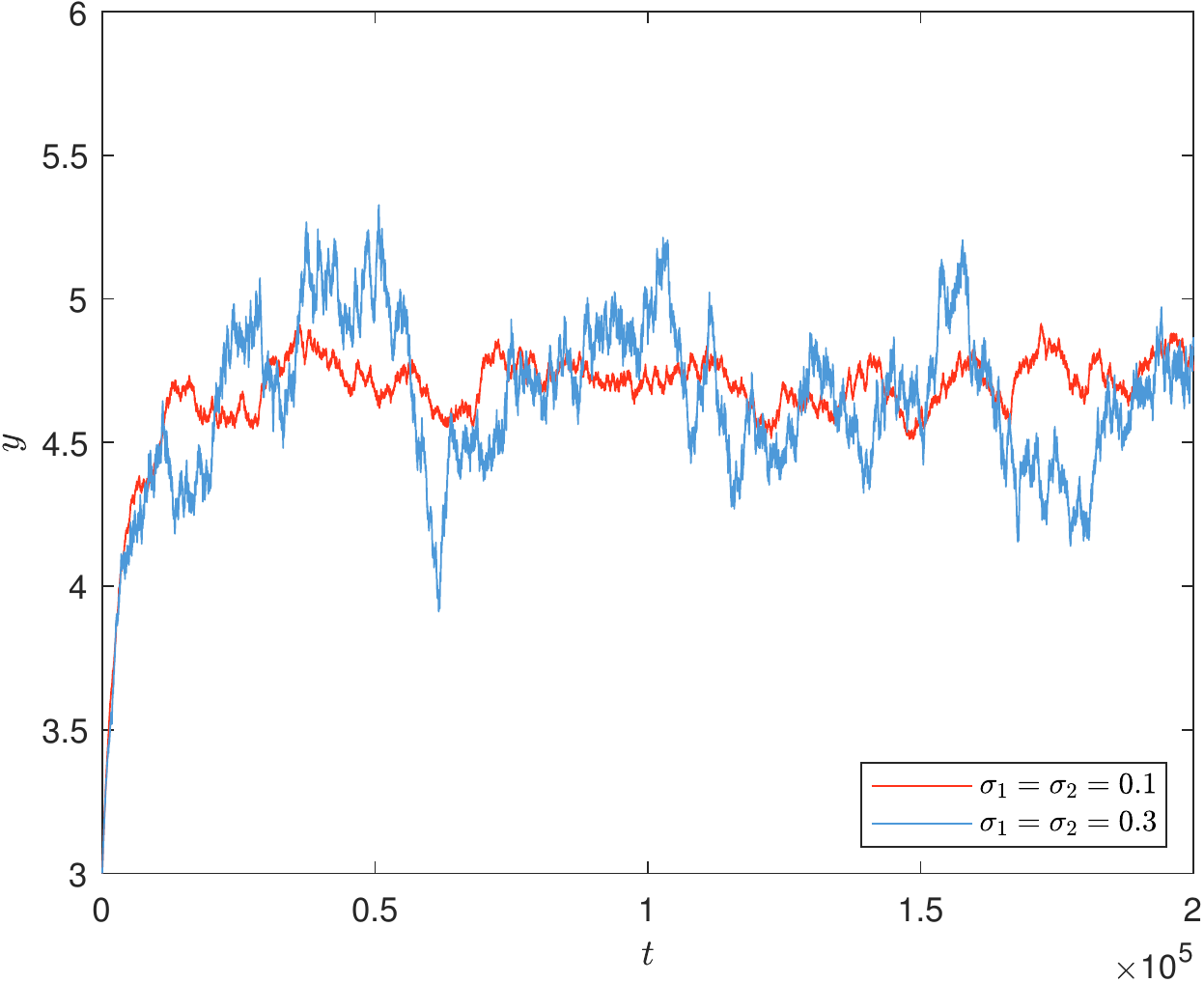}
			\caption{$ u_e=0.3, \beta=10 $}
			\label{fig:Th4y_3}
		\end{subfigure}
		\hfill
		\begin{subfigure}[b]{0.32\textwidth}
			\includegraphics[width=\textwidth]{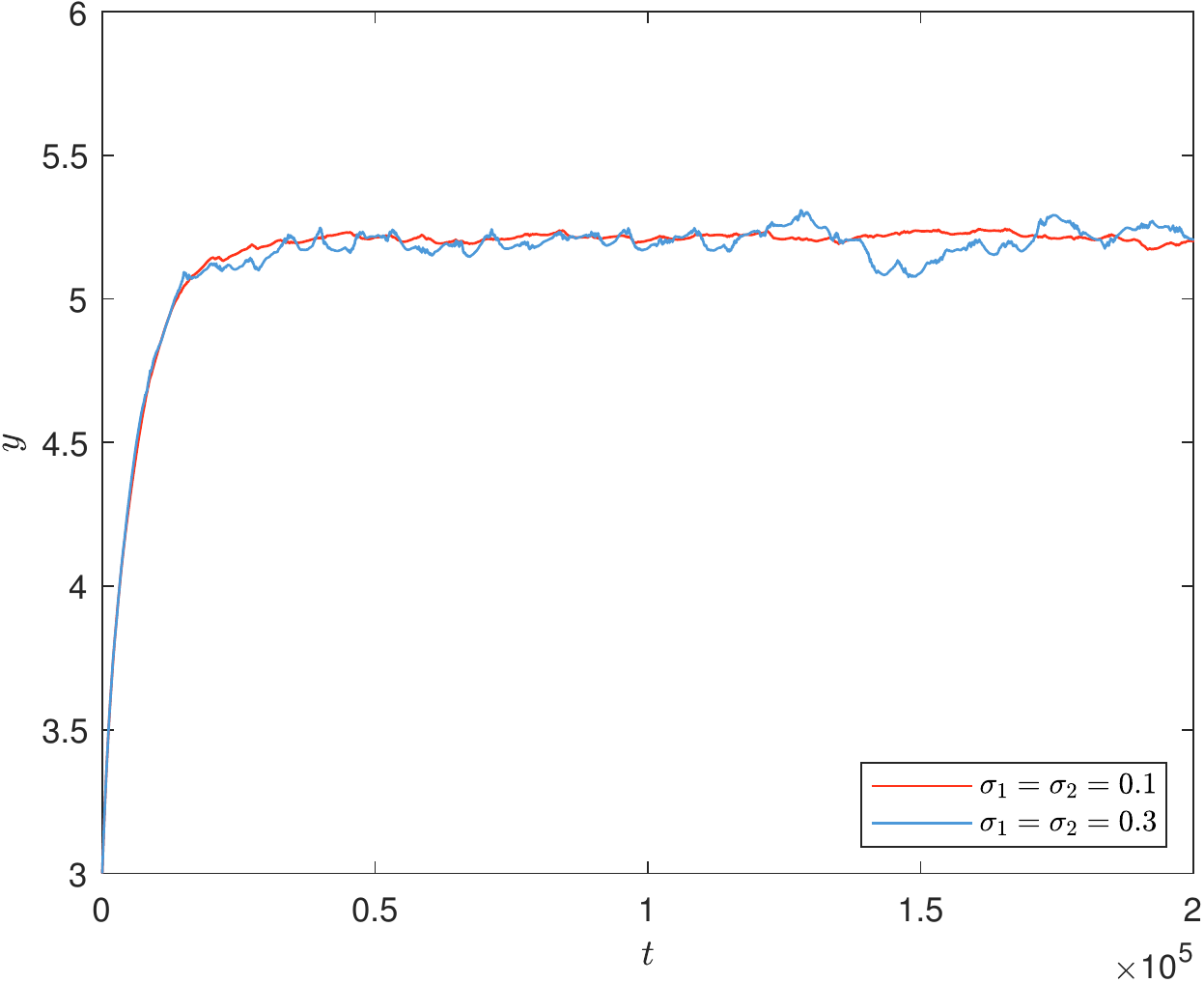}
			\caption{$ u_e=0.3+0.3\sin(t), \beta=0.1 $}
			\label{fig:Th4y_4}
		\end{subfigure}
		\hfill
		\begin{subfigure}[b]{0.32\textwidth}
			\includegraphics[width=\textwidth]{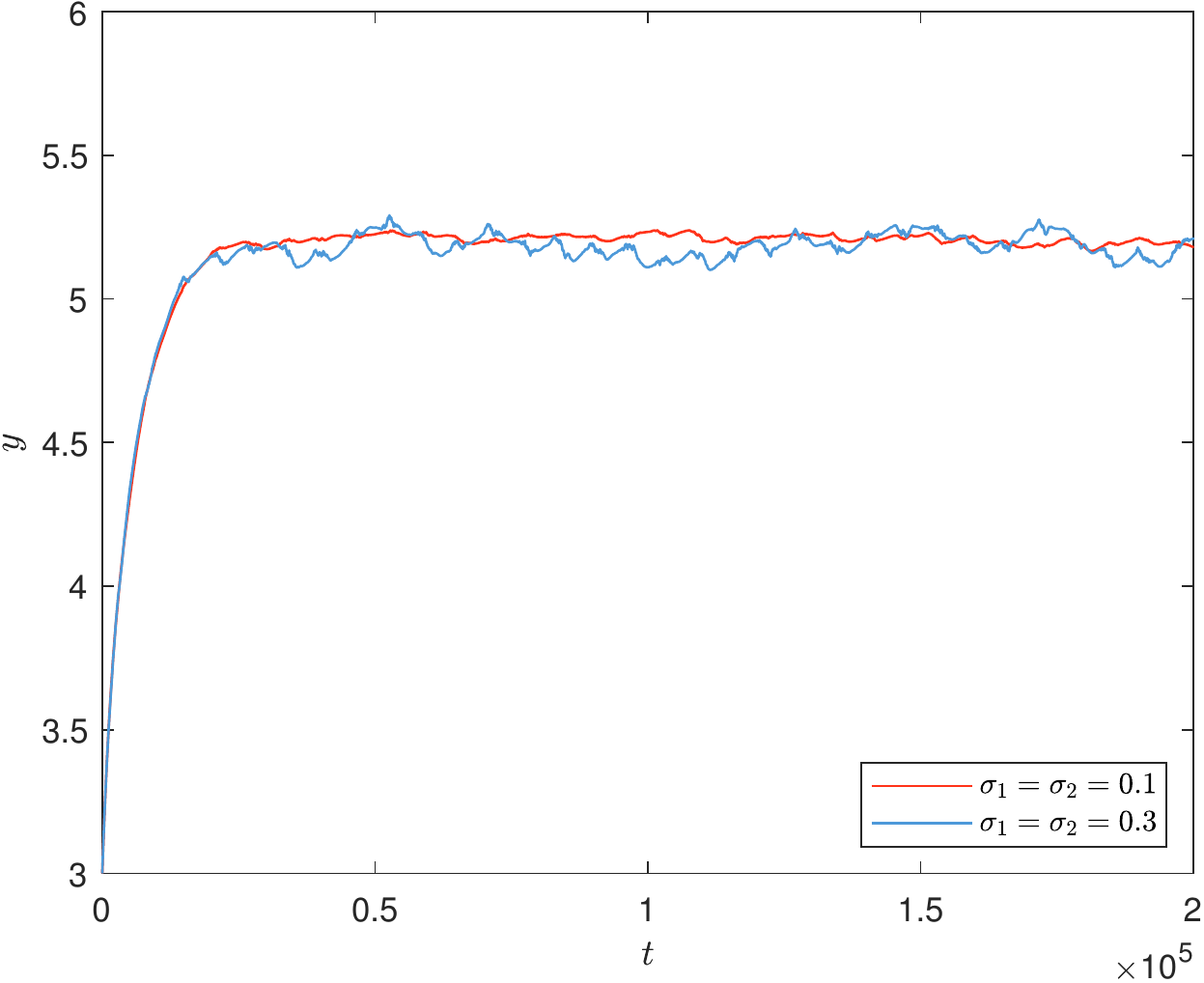}
			\caption{$ u_e=0.3+0.3\sin(t), \beta=1 $}
			\label{fig:Th4y_5}
		\end{subfigure}
		\hfill
		\begin{subfigure}[b]{0.32\textwidth}
			\includegraphics[width=\textwidth]{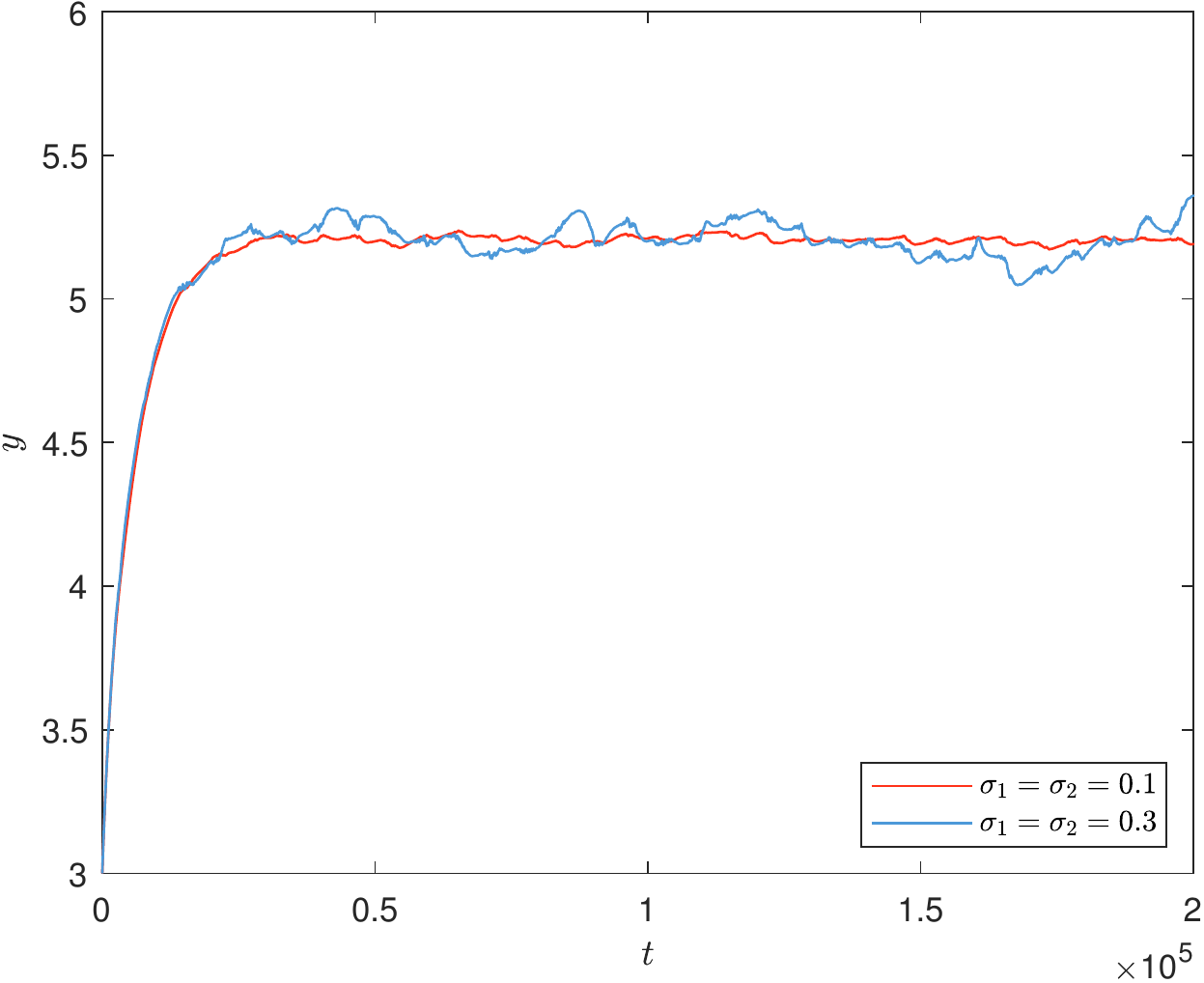}
			\caption{$ u_e=0.3+0.3\sin(t), \beta=10 $}
			\label{fig:Th4y_6}
		\end{subfigure}
		\caption{The density curves of $ y(t) $ with different $ u_e $ and $ \beta $}
		\label{fig:Th4y}
	\end{figure}

	\begin{figure}[H]
		\centering
		\begin{subfigure}{0.45\textwidth}
			\includegraphics[width=\textwidth]{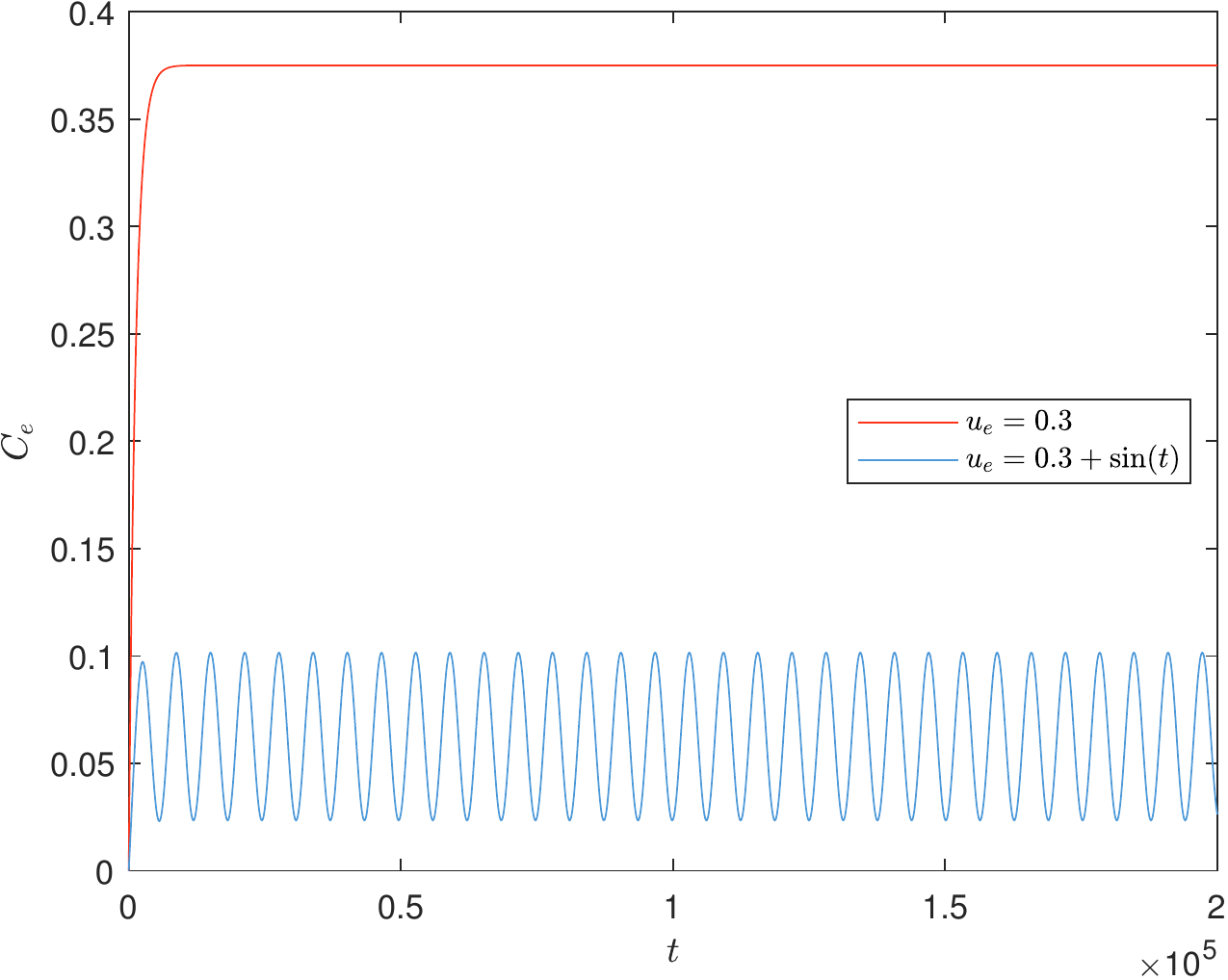}
			\caption{$ u_e=0.3$}
			\label{fig:Th4CeCo_1}
		\end{subfigure}
		\hfill
		\begin{subfigure}{0.45\textwidth}
			\includegraphics[width=\textwidth]{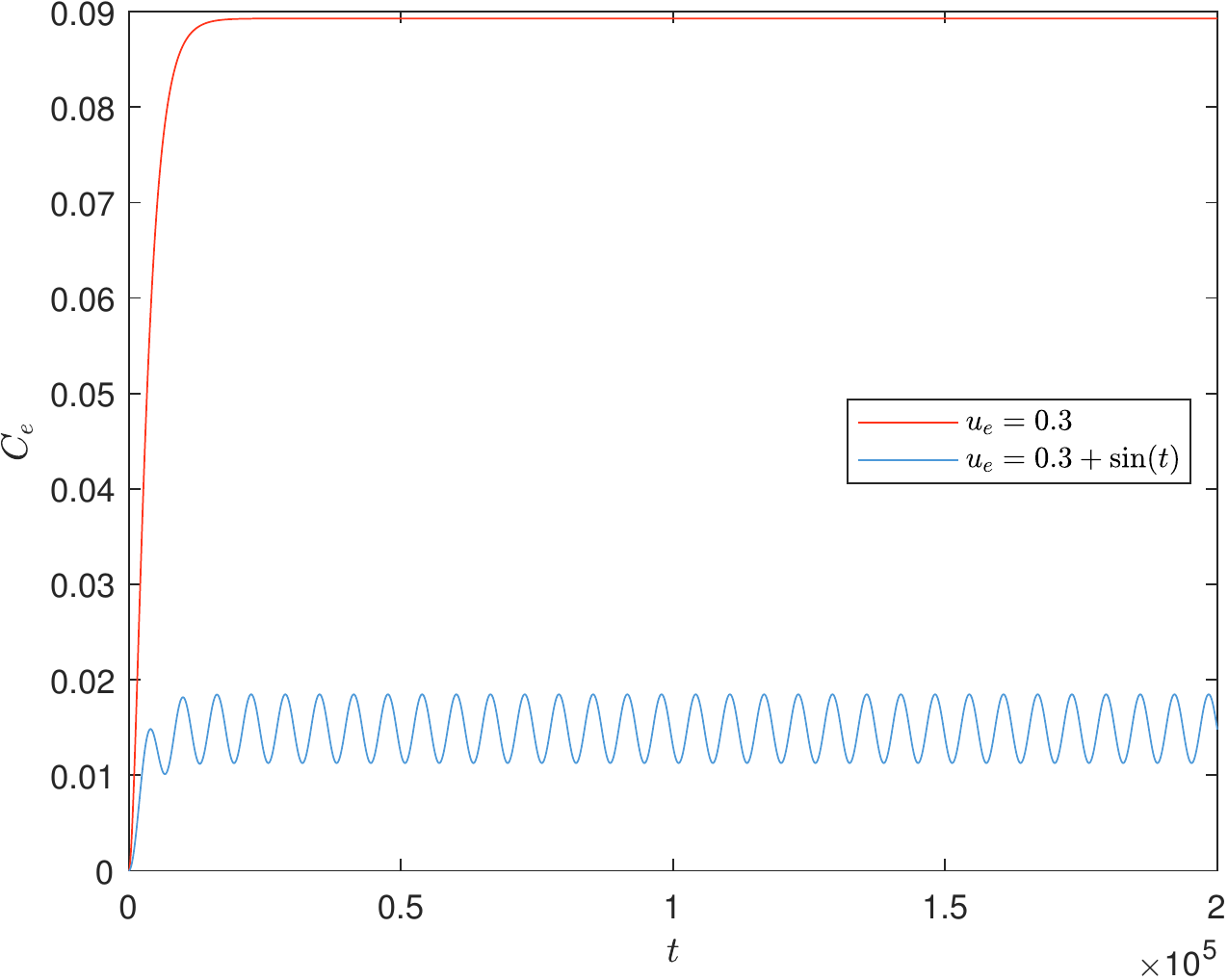}
			\caption{$ u_e=0.3+0.3\sin(t)$}
			\label{fig:Th4CeCo_2}
		\end{subfigure}
		\caption{Toxin concentration $ C_o(t) $ and $ C_e(t) $ with different $ u_e $.}
		\label{fig:Th4CeCo}
	\end{figure}

	\section{Discussion}
	\label{sec:discussion}

	The results of this paper can be further extended, especially for the single population model of such stage structure.
	There are many long-term behaviors that can be studied, and different methods can be used for theoretical practice.
	In reality, it can guide production and life in some fields, such as fishery production.
	When the production environment is polluted, it can be guided by relevant theories to take measures in time.
	At the same time, for the biological aquaculture which needs to separate the larva and adult, the subject can be developed into the theoretical basis for guiding the capture strategy~\cite{RN10}, and also can be introduced into the control theory.

	There is still space for further development for the two-stage structure model established in this paper.
	We assume that young individuals are transformed into adult individuals at a certain rate.
	In reality, the maturity process of any organism needs time accumulation, that is, the generation of time delay.
	In this regard, the model established in this paper also has certain limitations.
	Many scholars consider the existence of this factor, which leads to the time-delay term is used to describe the mature process~\cite{RN19,RN20,RN21}.

	\section{Conclusion}
	\label{sec:conclusion}

	In general, based on the previous studies, we introduce two-stage age structure, develop and expand some results.
	We develop the psychological effect function, select different psychological effect functions for different age structure, and extract some existing function types.

	In subsection~\ref{subsec:ode_model}, a single population model with age structure and psychological effects in polluted environment is established, which is a nonlinear time-varying system.
	Then in subsection~\ref{subsec:ode_stability}, we discuss the asymptotic stability of the system by Lyapunov first approximation theory, and give a sufficient condition for the stability.

	In subsection~\ref{subsec:sde_model}, based on subsection~\ref{subsec:ode_model}, the exposure rate of organisms to environmental toxins is affected by white noise, and then the exposure rate is modified into a random process, and the corresponding random single population model is established.
	The subsequent contents are mostly proved by Lyapunov function method.
	In subsection~\ref{subsec:sde_global_positive_solution}, we verify the existence of the globally unique positive solution of the stochastic model.
	In subsection~\ref{subsec:weakly_persistent}, near the non-pollution equilibrium point, we give the sufficient conditions for the weak mean persistence of a single population in the expected sense.
	In subsection~\ref{subsec:stochastic_permanence}, we give the sufficient conditions for the random persistence of a single population.

	In section~\ref{sec:simulations}, we make some numerical simulation to support our conclusions in Theorem~\ref{thm:weakly_persistent} and Theorem~\ref{thm:stochastic_permanence}.

	\section*{Conflict of Interest}
	\label{sec:Conflict of Interest}

	All authors stated that they have no conflict of interests in the paper.

	\section*{Acknowledgment}
	\label{sec:Acknowledgment}

	The research is supported by
	National Natural Science Foundation of China (61911530398),
	Special Projects of the Central Government Guiding Local Science and Technology Development (2021L3018),
	the Natural Science Foundation of Fujian Province of China (2021J01621),
	and Scientific Research Training Program in Fuzhou University (No.26040).

	\bibliographystyle{elsarticle-num}
	\biboptions{sort&compress}
	\bibliography{Wei_Wang_Yang}

\end{document}